\documentclass[reprint,amsmath,amssymb,aps,prx,floatfix]{revtex4-2}
\usepackage{graphicx}
\usepackage{dcolumn}
\usepackage{bm}
\usepackage{booktabs} 
\usepackage{braket}
\usepackage{multirow}
\usepackage{upgreek}
\usepackage{makecell}
\allowdisplaybreaks
\usepackage{xcolor}
\usepackage{algorithm}
\usepackage{algpseudocode}
\makeatletter
\expandafter\edef\csname ftype@algorithm\endcsname{\the\c@float@type}
\makeatother
\usepackage{amsthm}
\newtheorem{theorem}{Theorem}
\newtheorem{lemma}{Lemma}
\newtheorem{proposition}{Proposition}
\newtheorem{corollary}{Corollary}

\usepackage{hyperref}
\usepackage{orcidlink}
\hypersetup{colorlinks,linkcolor=blue,anchorcolor=blue,citecolor=blue,urlcolor=blue}
\DeclareMathOperator{\rank}{rank}
\DeclareMathOperator{\supp}{supp}
\newcommand{\R}{\mathbb{R}}
\newcommand{\Sp}{\mathrm{Sp}}
\newcommand{\Om}{\Omega}
\newcommand{\nmem}{n_{\mathrm{mem}}}

\algrenewcommand\algorithmiccomment[1]{\hfill\(\triangleright\)\ \textit{#1}}

\begin{document}

\title{\textbf{Memory-Optimal Sequential Synthesis of Multimode Gaussian Transformations}}%

\author{Fucheng Guo\,\orcidlink{0009-0003-0631-7404}}
\email{fguo22@ncsu.edu}
\affiliation{Department of Computer Science, North Carolina State University, Raleigh, North Carolina 27695, USA}%

\author{Frank Mueller\,\orcidlink{0000-0002-0258-0294}}
\email{fmuelle@ncsu.edu}
\affiliation{Department of Computer Science, North Carolina State University, Raleigh, North Carolina 27695, USA}%

\author{Yuan Liu\,\orcidlink{0000-0003-1468-942X}}
\email{q\_yuanliu@ncsu.edu}
\affiliation{Department of Electrical and Computer Engineering, North Carolina State University, Raleigh, North Carolina 27695, USA}%
\affiliation{Department of Computer Science, North Carolina State University, Raleigh, North Carolina 27695, USA}
\affiliation{Department of Physics and Astronomy, North Carolina State University, Raleigh, North Carolina 27695, USA}%

\date{\today}

\begin{abstract}
In modular quantum computing architectures, communication between hardware modules is mediated by traveling qumodes sent through transmission lines.
Each output qumode interacts with the emitting module only once through a
beam-splitter-type interaction and becomes inaccessible to that module after emission.
Information required for subsequent outputs must therefore remain in long-lived
memory qumodes. For a prescribed multimode Gaussian transformation on $N$
qumodes, this work determines the minimum memory cost for any given emission
order, constructs an explicit sequential protocol attaining this minimum, and
develops a greedy method for identifying memory-efficient emission orders. The
transformation is represented by a symplectic matrix $S$, specified either
directly or through a Gaussian gate sequence. The exact minimum memory cost is
obtained from the ranks of submatrices of $S$ and further reduces to a
support-based counting rule whose computational cost is linear in the size of the
support data. When $S$ is specified directly, a matrix-based protocol attains the minimum memory cost. If instead $S$ is specified through a gate sequence, the original gates can be reused without additional synthesis, although the resulting memory usage need not be minimal. Gaussian transformations with local support on a $D$-dimensional cubic lattice can be realized sequentially with $O(N^{(D-1)/D})$ memory qumodes. The protocols also apply to non-Gaussian inputs, including GKP
and cat states, and thereby provide an explicit, resource-efficient scheme for
intermodule communication in modular architectures for universal
continuous-variable quantum computation.
\end{abstract}

\maketitle

\section{Introduction}
\label{sec:intro}
Modular quantum architectures provide a scalable approach to building large quantum
systems from smaller, independently controlled modules
~\cite{Monroe2014,Nickerson2014}. Quantum states can be stored
and processed within individual modules and transmitted between them as traveling
qumodes through transmission lines~\cite{Kimble2008,Wehner2018,Reagor2016,Pfaff2017,Axline2018}.
In many applications, a module emits multiple traveling qumodes whose entanglement
structure is established by a multimode Gaussian transformation
~\cite{Menicucci2006,Menicucci2014,Yokoyama2013,Yoshikawa2016,
Asavanant2019,Larsen2019}. Here, Gaussian
refers to the operations used to establish entanglement among the qumodes, rather
than to the states carried by the individual qumodes locally. The input qumodes may therefore
carry arbitrary states, including non-Gaussian states such as
Gottesman--Kitaev--Preskill (GKP) states~\cite{GKP2001,Noh2020,Guo2026} and cat states
~\cite{Pfaff2017}. The resulting multimode
state may likewise be non-Gaussian, even though the entanglement among qumodes
is generated by Gaussian operations~\cite{Noh2020,Guo2026,Pfaff2017}.

Consider a general multimode Gaussian transformation acting on $N$ input qumodes.
In phase space, the transformation is represented by a symplectic matrix
$S\in\Sp(2N,\R)$ that maps the input quadratures to those of the $N$ output
qumodes~\cite{BraunsteinVanLoock2005,Weedbrook2012,4rf7-9tfx}. The transformation can be
decomposed into phase-space rotations, single-qumode squeezing operations, and
beam-splitter (BS) interactions
~\cite{Weedbrook2012,Chakhmakhchyan2018}. These operations can be implemented on stationary qumodes and controlled on
demand~\cite{Gao2018,Chapman2023}.

The output qumodes are emitted sequentially as traveling qumodes. Depending on the
physical platform, stored states may be released through tunable coupling to a
transmission line or through parametric mode conversion, while spontaneous and
drive-assisted emission are commonly used to generate traveling qumodes directly
~\cite{Yin2013,Pierre2014,Pfaff2017,Houck2007,Pechal2014}.
The present work considers a state-independent BS-type release, in
which a BS interaction transfers the state of a stationary qumode to a traveling qumode
through parametrically controlled mode conversion
~\cite{Pfaff2017,Axline2018}. After the transfer, the emitted
qumode is no longer accessible to the emitting module. The stationary qumodes
retained in the module and reused throughout the sequential implementation are
referred to here as memory qumodes. For a given multimode Gaussian transformation,
this loss of access raises three central questions: (1) How can an emission order
requiring few memory qumodes be found? (2) For a chosen emission order, what is the
minimum number of memory qumodes required? (3) How can a sequential protocol
achieving this minimum be constructed?

Previous studies of sequential generation have mainly focused on preparing a
prescribed target state~\cite{Schon2005,Schon2007}. Matrix-product-state (MPS) methods describe such processes
in terms of the dimension of an auxiliary system, with analogous formulations
available for Gaussian states~\cite{Schon2005,Schon2007,Adesso2006}. Related work has also studied the minimum number of
emitters required to generate photonic graph states
~\cite{LindnerRudolph2009,Li2022}. Large-scale optical experiments
have further demonstrated the sequential generation of specific one- and
two-dimensional continuous-variable cluster states
~\cite{Yokoyama2013,Yoshikawa2016,Asavanant2019,Larsen2019}, establishing the scalability
of this approach for particular families of target states.

The problem considered here is different. MPS-based methods may characterize the
sequential generation of an individual output state
~\cite{Schon2005,Schon2007,Adesso2006}, but they do not by themselves
characterize the sequential implementation of a prescribed multimode Gaussian
transformation for arbitrary input states. In the present setting, the relevant
resource is the number of reusable memory qumodes required to preserve the
information needed for outputs that have not yet been emitted. Existing
state-preparation results therefore do not determine the minimum memory cost of
implementing a general symplectic transformation, account for its dependence on the
emission order, identify an emission order with low memory cost, or construct a
sequential protocol guaranteed to attain the minimum
~\cite{Schon2005,Schon2007,Adesso2006,LindnerRudolph2009,Li2022,
Yokoyama2013,Yoshikawa2016,Asavanant2019,Larsen2019}.

This work contributes an analysis that establishes how to determine the minimum memory cost for any prescribed
emission order and constructs a sequential protocol that attains this minimum. It
also develops an efficient method for identifying emission orders with low memory
cost. These results are demonstrated through two applications: continuous-variable
quantum error-correcting encoders, which map logical information into multimode code
spaces that enable error detection and correction
~\cite{Braunstein1998,LloydSlotine1998,Noh2020,Guo2026}, and the sequential generation of
continuous-variable cluster states, which provide a key resource for
measurement-based quantum computation
~\cite{Menicucci2006,Menicucci2014}.

The main results of this work are as follows. First, for any
$S\in\Sp(2N,\R)$ and any prescribed emission order, the minimum number of memory
qumodes is given by one-half of the maximum rank of a corresponding family of
submatrices of $S$ (Sec.~\ref{sec:mincav}). This rank expression further reduces to
a simple counting rule based on the support pattern of $S$. Consequently, the
minimum memory cost can be evaluated in time linear in the size of the support
data, without explicit rank calculations. Second, two sequential synthesis
protocols are developed. The gate-sequence protocol presented in
Sec.~\ref{sec:algorithm} reuses the gates of a given Gaussian circuit without
additional synthesis, although its memory usage need not be minimal. The
matrix-based protocol in Appendix~\ref{app:matrix} constructs new in-register
Gaussian operations from the target symplectic matrix and always achieves the
minimum memory cost. Third, the memory cost depends strongly on the emission order.
A deterministic greedy heuristic is developed to identify memory-efficient
emission orders (Secs.~\ref{sec:chain} and~\ref{sec:order}). Finally, Gaussian transformations with local support on a $D$-dimensional cubic lattice can be realized sequentially with $O(N^{(D-1)/D})$ memory qumodes (Sec.~\ref{sec:scaling}).

\begin{figure*}[t]
  \centering
  \includegraphics[width=0.7\textwidth]{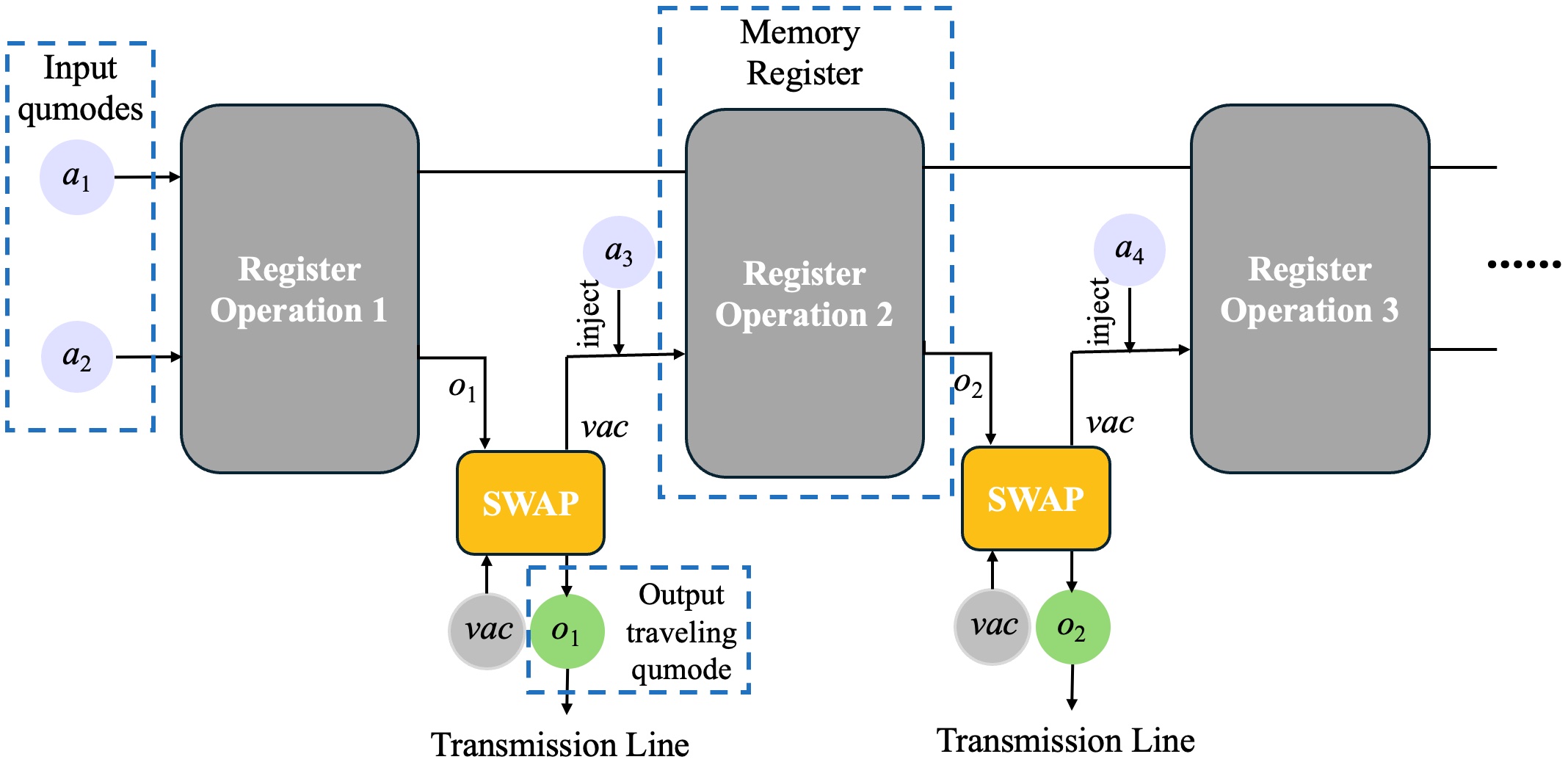}
 \caption{Sequential implementation of a multimode Gaussian transformation $S$
  using a reusable memory register. The memory register consists of stationary
  qumodes on which Gaussian operations can be performed. Time proceeds from
  left to right. The initial input qumodes are first loaded into the register.
  At each step, an in-register Gaussian operation maps the next output qumode
  $o_k$ onto one memory qumode while retaining the information required for later
  outputs in the remaining memory qumodes. The state of $o_k$ is then transferred
  to a vacuum traveling qumode through a beam-splitter full swap
  ($\theta=\pi$, labeled ``SWAP'') and released into a transmission line. The
  swap leaves the corresponding memory qumode in the vacuum state, which can then
  be replaced by a newly injected input qumode and reused in the next register
  operation. Repeating this procedure emits the outputs sequentially in the
  prescribed order $o_1,\ldots,o_N$. After emission, each traveling qumode is no
  longer accessible to the emitting module.}
  \label{fig:model}
\end{figure*}

\section{Model and Setup}
\label{sec:model}

\subsection{Gaussian Transformation and Sequential Architecture}

A Gaussian unitary acting on $N$ bosonic qumodes is represented in phase space by a
symplectic matrix $S\in\Sp(2N,\R)$~\cite{Noh2020,Menicucci2014}. Defining the
quadrature vector
\begin{equation}
\hat{\mathbf{r}}
=
(\hat q_1,\hat p_1,\dots,\hat q_N,\hat p_N)^{\top},
\label{eq:quadvector}
\end{equation}
with $[\hat q_j,\hat p_k]=i\,\delta_{jk}$ and
$[\hat q_j,\hat q_k]=[\hat p_j,\hat p_k]=0$, the transformation is
\begin{equation}
  \hat{\mathbf{r}} \mapsto S\,\hat{\mathbf{r}},
  \qquad
  S^{\top}\Om S=\Om,
  \qquad
  \Om=\bigoplus_{j=1}^{N}
  \begin{pmatrix}
    0 & 1\\
    -1 & 0
  \end{pmatrix}.
  \label{eq:symplectic}
\end{equation}
The target Gaussian transformation is specified by $S$. The input states are
determined by the chosen code and are not included in $S$.

To implement $S$, the $N$ output qumodes are generated sequentially
~\cite{Pichler2017}. They are produced using a small set of stationary, reusable
memory qumodes on which arbitrary Gaussian operations can be performed
~\cite{Heeres2017,Gao2018}. Each output is transferred to a separate traveling
qumode that is initially in the vacuum state and interacts once, through a full
swap, with the memory qumode carrying that output before leaving the module. As
shown schematically in Fig.~\ref{fig:model}, the initial input qumodes are first
loaded into the memory register. At each step, Gaussian operations within the
register map the next output $o_k$ onto one memory qumode while retaining the
information required for later outputs in the remaining memory qumodes. A
beam-splitter full swap then transfers $o_k$ to the vacuum traveling qumode and
leaves the corresponding memory qumode in the vacuum state. When additional inputs
are required, new input qumodes can be injected into available memory qumodes before
the next register operation. The number and timing of these injections may vary from
step to step according to the target transformation and emission order.

\subsection{Emission of Traveling Qumodes}

The release is implemented by a BS interaction. As a concrete realization, in circuit quantum electrodynamics (QED)~\cite{Reagor2016,Pfaff2017,Axline2018},
driving the Josephson nonlinearity at the frequency difference between a memory qumode
$a$ and a traveling qumode $b$ gives the rotating-frame Hamiltonian
\begin{equation}
  H(t)/\hbar = g(t)\,\hat a^{\dagger}\hat b
  + g^{*}(t)\,\hat a\,\hat b^{\dagger},
  \label{eq:Hbs}
\end{equation}
where $\hat a$ and $\hat b$ are the corresponding annihilation operators. The
complex coupling amplitude $g(t)$ is controlled by the strength and phase of the
parametric pump and can be switched on and off on demand
~\cite{Pfaff2017,Gao2018,Chapman2023}. Since the phase of $g(t)$ can be absorbed into
the definition of the traveling qumode, the corresponding beam-splitter gate may be
written as
\begin{equation}
\label{eq:Ubs}
\begin{aligned}
\hat U_{\mathrm{BS}}(\theta)
&=
\exp\!\left[
-i\frac{\theta}{2}
\left(
\hat a^\dagger \hat b
+
\hat a\hat b^\dagger
\right)
\right],\\
\theta
&=
2\int_0^{\tau} \lvert g(t)\rvert\,\mathrm{d}t ,
\end{aligned}
\end{equation}
where $\tau$ is the duration of the beam-splitter pulse.

In this work, each emission is required to transfer the state of the memory qumode
completely to the traveling qumode and leave the memory qumode in the incoming vacuum
state. This is achieved by choosing $\theta=\pi$, for which the beam-splitter gate
implements a full swap. Because Eq.~\eqref{eq:Hbs} is bilinear in the annihilation
operators, the swap implements a fixed symplectic transformation that is independent
of the state being released~\cite{Pfaff2017}. Thus, each traveling qumode interacts
with the memory register exactly once and, after emission into the transmission line,
can no longer be operated on by the emitting module.

\subsection{Memory Cost for a Fixed Emission Order}

Because emitted qumodes cannot be accessed again, the emission order affects the
required memory. Let $S[o,a]$ denote the $2\times2$ block of $S$ that maps input
qumode $a$ to output qumode $o$. The support of output $o$ is defined as
\begin{equation}
  \supp(o) = \{\,a : S[o,a]\neq 0\,\}.
  \label{eq:support}
\end{equation}

For the exact implementation considered here, the support is determined solely by
whether or not a block is zero as every nonzero block $S[o,a]$ is included in
$\operatorname{supp}(o)$, regardless of its magnitude. As a result, even very small
nonzero blocks can enlarge the support and increase the predicted memory cost.
When many such blocks are present, an approximate implementation may therefore
provide a more practical description. One possible approach is to choose a matrix
norm and a threshold $\varepsilon>0$, and define
\begin{equation}
\operatorname{supp}_{\varepsilon}(o)
=
\left\{
a:\|S[o,a]\|>\varepsilon
\right\},
\end{equation}
so that blocks satisfying $\|S[o,a]\|\leq\varepsilon$ are neglected. This thresholded
support may reduce the required memory, but it also perturbs the target transformation.
A complete treatment would therefore require quantifying the resulting implementation
error, determining how it accumulates across the protocol, and analyzing its dependence
on the input state. Such an approximate formulation is beyond the scope of the present
work, which focuses on exact implementations.

To implement $S$ faithfully, every input in $\supp(o)$ must be present in the
register before $o$ is emitted. Thus, when the $k$-th output $o_k$ is released, the
set of injected inputs is
\begin{equation}
  P_k = \bigcup_{j\le k}\supp(o_j),
  \label{eq:Pk}
\end{equation}
while the outputs that remain to be emitted are
$F_k=\{o_k,\dots,o_N\}$.

Although inputs could be injected earlier, doing so cannot reduce the required
memory. A memory-optimal protocol therefore injects each input only when it is first
needed, so Eq.~\eqref{eq:Pk} gives the injected set at every step. The correlations
that must remain in the register at the $k$-th emission are described by the
submatrix $S[F_k;P_k]$, whose rows correspond to the quadratures of the outputs in
$F_k$ and whose columns correspond to those of the inputs in $P_k$. Since these
inputs cannot be injected again and these outputs have not yet been released, the
memory register is the only system that can retain these correlations. The rank of
$S[F_k;P_k]$ counts the number of independent quadratures involved, and the factor
of $\tfrac12$ converts these quadratures into conjugate pairs, one pair per memory
qumode. Thus, $\tfrac12\rank S[F_k;P_k]$ is the natural candidate for the number of
memory qumodes required at this step. The next section proves that this rank is even
and that this candidate is exact.
The largest value of this quantity
over the emission sequence is therefore the natural candidate for the memory cost
of implementing $S$ in the chosen order.

Section~\ref{sec:mincav} shows that this candidate is exact: No faithful
implementation can use fewer memory qumodes, and an explicit protocol achieves the
bound. The same section also reduces the rank expression to a simple count, namely
the number of inputs injected so far minus the number of outputs already emitted.

\section{Minimum Number of Memory Qumodes}
\label{sec:mincav}

Section~\ref{sec:model} showed that, before the $k$-th output is emitted, the memory
register must preserve the correlations described by the submatrix $S[F_k;P_k]$.
These correlations require at least
$\tfrac12\rank S[F_k;P_k]$ memory qumodes at that step. Therefore, the largest of
these quantities over the full emission sequence is a lower bound on the memory cost
for the chosen emission order. It is not immediate that this lower bound can actually
be achieved, because the memory must preserve the correct information not only for the
current output but also for all later outputs. The following theorem proves that a
single sequential protocol can always achieve this bound. A subsequent corollary
reduces the rank expression to a simple count determined only by the support pattern
of $S$.

\begin{theorem}[Minimum number of memory qumodes]\label{thm:ncav}
For a given emission order $o_1,\dots,o_N$, the minimum number of memory qumodes
required to realize $S\in\Sp(2N,\R)$ faithfully under sequential emission by single
full swaps is
\begin{equation}
  n_{\mathrm{mem}}(S)
  =
  \max_{1\le k\le N}\tfrac12\rank S[F_k;P_k],
  \label{eq:ncav}
\end{equation}
where
$P_k=\bigcup_{j\le k}\supp(o_j)$ and
$F_k=\{o_k,\dots,o_N\}$.
Each $\rank S[F_k;P_k]$ is even, and therefore $n_{\mathrm{mem}}(S)$ is an integer.
\end{theorem}

The proof consists of two parts. The lower bound shows that any faithful sequential
implementation must use at least the number of memory qumodes in
Eq.~\eqref{eq:ncav}. The upper bound gives an explicit construction that achieves
this value.

\subsection{Lower Bound}
\label{sec:mincav-lower}
At each emission step, any correlations between the inputs already injected and the outputs not yet emitted must be preserved in the memory register, which gives the lower bound on the required memory.
To derive the lower bound, we describe the input and output quadratures in a common
space,
$V=\R^{2N}_{\mathrm{in}}\oplus\R^{2N}_{\mathrm{out}}$,
where the two components correspond to the input and output phase spaces. We equip
$V$ with the symplectic form
$\Om_V=\Om_{\mathrm{in}}\oplus(-\Om_{\mathrm{out}})$.
The transformation $S$ is represented by its graph,
\begin{equation}
  L=\{(\mathbf r,S\mathbf r):\mathbf r\in\R^{2N}\},
  \qquad \dim L=2N.
  \label{eq:graph}
\end{equation}
The minus sign in $\Om_V$ makes $L$ a Lagrangian subspace. Indeed, for any
$(\mathbf r,S\mathbf r),(\mathbf r',S\mathbf r')\in L$,
$\mathbf r^{\top}\Om\mathbf r'
-(S\mathbf r)^{\top}\Om(S\mathbf r')=0$
because $S^{\top}\Om S=\Om$, and $\dim L=\tfrac12\dim V$.
This representation is the Gaussian analogue of the Choi construction and allows
the input--output correlations to be analyzed across the emission sequence
~\cite{Schon2005,Schon2007,Adesso2006,Eisert2010}.

Consider the cut immediately before the $k$-th emission. The future subspace
contains the remaining inputs $\bar P_k$ and the outputs
$F_k=\{o_k,\dots,o_N\}$ that have not yet been emitted. The past subspace $W$
contains the injected inputs $P_k$ and the previously emitted outputs
$\bar F_k=\{o_1,\dots,o_{k-1}\}$, the complement of $F_k$. Because each input
is injected before its first use, $S[\bar F_k;\bar P_k]=0$ by
Eqs.~\eqref{eq:support} and~\eqref{eq:Pk}.

The number of qumodes shared across this cut is
\begin{equation}
  b_k=\tfrac12\bigl[\dim\pi_W(L)-\dim(L\cap W)\bigr],
  \label{eq:bond}
\end{equation}
where $\pi_W(L)$ is the projection of $L$ onto the past subspace, while $L\cap W$
is the part of $L$ contained entirely in the past. Their dimension difference counts
the past quadrature degrees of freedom that remain connected to the future; the
factor $1/2$ converts quadrature pairs into qumodes. This is the Gaussian analogue of
the bond dimension in a matrix-product-state description
~\cite{Schon2005,Adesso2006}.

Ordering the inputs as $(P_k,\bar P_k)$ and the outputs as $(\bar F_k,F_k)$ gives
\begin{equation}
  S=
  \begin{pmatrix}
    S[\bar F_k;P_k] & 0\\[2pt]
    S[F_k;P_k] & S[F_k;\bar P_k]
  \end{pmatrix}.
  \label{eq:Sblock}
\end{equation}
For $(\mathbf r,S\mathbf r)\in L$, the past component is
$\bigl(\mathbf r_{P_k},S[\bar F_k;P_k]\mathbf r_{P_k}\bigr)$.
Since $\mathbf r_{P_k}\in\R^{2|P_k|}$ is arbitrary,
$\dim\pi_W(L)=2|P_k|$. Here $|\cdot|$ denotes the number of elements of a set.

A point of $L$ lies entirely in $W$ when
$\mathbf r_{\bar P_k}=0$ and
$S[F_k;P_k]\mathbf r_{P_k}=0$.
Hence,
$\dim(L\cap W)=2|P_k|-\rank S[F_k;P_k]$.
Substituting these dimensions into Eq.~\eqref{eq:bond} gives
\begin{equation}
  b_k=\tfrac12\rank S[F_k;P_k].
  \label{eq:bk}
\end{equation}
The symplectic Schmidt decomposition groups the shared quadratures into conjugate
pairs~\cite{Botero2003,Adesso2006}, so the rank is even and $b_k$ is an integer.

At this cut, the memory register is the only system connecting the injected inputs
$P_k$ to the outputs $F_k$ that remain to be emitted. A faithful implementation must
therefore retain at least $b_k$ memory qumodes. Maximizing over the sequence yields
\begin{equation}
  n_{\mathrm{mem}}(S) \ge \max_k b_k
  = \max_k \tfrac12\rank S[F_k;P_k].
  \label{eq:lb}
\end{equation}

\subsection{Upper Bound}
\label{sec:mincav-upper}

A sequential construction is given that realizes $S$ using at most
$\max_k \tfrac12\rank S[F_k;P_k]$ memory qumodes, thereby matching the lower bound
derived above.

The rows of $S$ form a symplectic frame. Since $S\Om S^{\top}=\Om$,
\begin{equation}
  \{o_i,o_j\}_{\Om}
  =
  \delta_{ij}
  \begin{pmatrix}
    0 & 1\\
    -1 & 0
  \end{pmatrix}.
  \label{eq:symframe}
\end{equation}
Thus, the two quadratures of each output qumode satisfy the canonical commutation
relation, while quadratures belonging to different output qumodes commute.

Consider the $k$-th emission step. The inputs in $P_k$ have already been injected,
and the outputs $o_1,\dots,o_{k-1}$ have already been emitted. Since
$\supp(o_k)\subseteq P_k$, the output $o_k$ depends only on inputs already loaded
into the device. It also commutes with every previously emitted output by
Eq.~\eqref{eq:symframe}. The injected input space can be decomposed into the emitted
output qumodes and the qumodes still held in the register. Because $o_k$ has no
component along the emitted qumodes, it lies in the span of the current register
qumodes.

A Gaussian operation within the register can therefore map $o_k$ onto a single
memory qumode without introducing an additional qumode. A $\pi$ beam-splitter full
swap then transfers this state to a vacuum traveling qumode, which is released as
$o_k$. Equation~\eqref{eq:symframe} also implies that $o_k$ commutes with all later
outputs, so removing it does not disturb the information required for the remaining
emissions. Repeating this procedure produces all outputs of $S$ in the prescribed
order.

It remains to bound the register size. Immediately before the $k$-th emission,
$|P_k|$ inputs have been injected and $k-1$ outputs have been emitted, so the
register contains $|P_k|-(k-1)$ memory qumodes. The restriction of $S$ to the columns
indexed by $P_k$ has full column rank $2|P_k|$, since these columns are a subset of
the linearly independent columns of $S$. The $2(k-1)$ rows associated with the
emitted outputs are supported entirely on $P_k$ and remain linearly independent.
The remaining rows are those of $S[F_k;P_k]$, and therefore
$\rank S[F_k;P_k]\ge 2|P_k|-2(k-1)$. Hence,
$|P_k|-(k-1)\le \tfrac12\rank S[F_k;P_k]$.

Taking the maximum over the emission sequence gives
\begin{equation}
  \nmem(S)
  \le
  \max_k \tfrac12\rank S[F_k;P_k].
  \label{eq:ub}
\end{equation}

Together with the lower bound in Eq.~\eqref{eq:lb}, this proves
Eq.~\eqref{eq:ncav}.
Appendix~\ref{app:matrix} gives the construction as an explicit
algorithm.

\subsection{Counting Form}
\label{sec:mincav-count}

The rank expression in Eq.~\eqref{eq:ncav} can be evaluated without performing any
rank computation.

\begin{corollary}[Counting form of the minimum]\label{cor:count}
For every $S\in\Sp(2N,\R)$, every emission order, and every step $k$,
\[
  \tfrac12\rank S[F_k;P_k]=|P_k|-(k-1).
\]
Therefore,
\begin{equation}
  \nmem(S)
  =
  \max_{1\le k\le N}
  \bigl(|P_k|-k+1\bigr) .
  \label{eq:count}
\end{equation}

\end{corollary}

\begin{proof}
The upper-bound argument established
$\rank S[F_k;P_k]\ge 2|P_k|-2(k-1)$.

For the opposite inequality, the previously emitted outputs are supported on the
inputs in $P_k$. By Eq.~\eqref{eq:symframe}, they span a
$2(k-1)$-dimensional symplectic subspace of the $2|P_k|$-dimensional phase space of
the injected inputs. The rows of $S[F_k;P_k]$ are symplectically orthogonal to this
subspace, since they correspond to outputs different from the emitted ones. They
must therefore lie in its symplectic complement, whose dimension is
$2|P_k|-2(k-1)$. Hence,
$\rank S[F_k;P_k]\le 2|P_k|-2(k-1)$, which proves the equality.
\end{proof}

Equation~\eqref{eq:count} depends only on the support sets defined in
Eq.~\eqref{eq:support}, not on the numerical values of the entries of $S$. Thus, the
memory cost is determined by which outputs depend on which inputs. In particular,
changing the squeezing parameters can change the entanglement across an emission cut
without changing the memory cost, provided that the support pattern remains
unchanged. The memory requirement therefore counts the number of qumodes that must be
retained, rather than the amount of entanglement across the cut.

In practice, the cost is obtained by updating the running union
$P_k=\bigcup_{j\le k}\supp(o_j)$ and evaluating $|P_k|-k+1$ at each step. This
requires
$O\!\left(\sum_k |\supp(o_k)|\right)$
time, which is linear in the size of the support data and involves no linear algebra.

\section{Sequential Generation Algorithm}
\label{sec:algorithm}

The target transformation $S$ may be specified either as a symplectic matrix or as
a Gaussian gate sequence. When only the symplectic matrix is available, the
matrix-based protocol in Appendix~\ref{app:matrix}, discussed further in
Sec.~\ref{sec:general}, constructs the required in-register Gaussian operations and
achieves the minimum memory cost given by Theorem~\ref{thm:ncav}. This section
considers the second case, in which $S$ is provided as a Gaussian gate sequence,
\begin{equation}
  S = G_L\cdots G_1,
  \label{eq:decomposition}
\end{equation}
where each gate $G_l$ acts on the set of qumodes $\supp(G_l)$. This is the natural
representation for many applications, since bosonic encoding circuits and
cluster-state constructions are usually specified as circuits rather than as
symplectic matrices.

When a gate sequence is available, the Gaussian operations applied within the
memory register can be taken directly from that sequence and scheduled according
to the chosen emission order, without synthesizing new operations from $S$. The
resulting protocol faithfully realizes the target transformation, although its
memory usage need not attain the minimum given by Theorem~\ref{thm:ncav}.

\subsection{Entanglement Structure of the Gate Sequence}
\label{sec:entstructure}

The protocol is described in the Heisenberg picture. Each memory qumode is represented
by a $2\times 2N$ row block that expresses its two quadratures in terms of the input
quadratures. When a qumode is injected, its row block is the identity on the
corresponding input qumode. Here, $\supp(G_l)$ denotes the fixed set of qumodes on
which the gate $G_l$ acts directly. Applying $G_l$ updates only the row blocks of
the qumodes in $\supp(G_l)$, replacing them by linear combinations of those row
blocks, while all other row blocks remain unchanged. These updated row blocks may
already contain contributions from other input qumodes introduced through earlier
entangling gates. Consequently, although $G_l$ acts directly only on
$\supp(G_l)$, its effect can propagate through the gate sequence and influence the
final row blocks of additional qumodes.

This dependence is described by the entanglement support $E(q)$ of qumode $q$. It is
defined as the smallest set of gate indices satisfying the following two conditions:
It contains every $l$ for which $q\in\supp(G_l)$, and whenever it contains $l$, it
also contains every earlier index $l'<l$ such that
$\supp(G_{l'})\cap\supp(G_l)\neq\emptyset$. Thus, $E(q)$ contains all gates whose
effects can propagate to $q$ through the gate sequence. Any index set with this
closure property is called support-closed. For an index set $F$, let $\Pi_F$ denote
the product of the gates $\{G_l:l\in F\}$ taken in the same order as in
Eq.~\eqref{eq:decomposition}.

\begin{lemma}[Locality of the entanglement support]
\label{lem:entsupport}
Let $F$ be support-closed and satisfy $F\supseteq E(q)$. Then the row pair of qumode
$q$ in $\Pi_F$ is identical to its row pair in the full transformation $S$.
\end{lemma}

\begin{proof}
The proof is by induction on the sequence length $L$. The statement is immediate for
$L=0$, since both transformations are the identity. Assume it holds for sequences of
length $L-1$, and write
$S=G_L\Pi'$, where $\Pi'=G_{L-1}\cdots G_1$.

If $q\notin\supp(G_L)$, then $L\notin E(q)$ and $G_L$ does not change the row pair of
$q$. The set $F\setminus\{L\}$ remains support-closed for the shortened sequence and
contains the entanglement support of $q$ in that sequence. The result then follows
from the induction hypothesis.

\begin{algorithm}[t]
\caption{Sequential generation of $S$ from a gate sequence}
\label{alg:slicing}
\begin{algorithmic}[1]

\Statex \textbf{Definitions:}
\Statex $G_l$: $l$-th gate in $S=G_L\cdots G_1$.
\Statex $\supp(G_l)$: qumodes acted on by $G_l$.
\Statex $W$: qumode labels reached by the reverse scan, initialized as
$W=\{o_k\}$.
\Statex $\mathcal{C}$: ordered indices of unapplied gates selected by the scan.
\Statex $\mathcal{G}$: sequential operation sequence constructed by the algorithm.

\Statex \rule{\linewidth}{0.4pt}

\Statex \textbf{Input:}
\Statex Gate sequence $S=G_L\cdots G_1$.
\Statex Emission order $o_1,\dots,o_N$.

\Statex \rule{\linewidth}{0.4pt}

\Statex \textbf{Output:}
\Statex Operation sequence $\mathcal{G}$ emitting outputs in the order
$o_1,\dots,o_N$.

\Statex \rule{\linewidth}{0.4pt}

\State $\mathcal{G}\gets[\,]$
\State mark all $G_l$ as unapplied

\For{$k=1,\dots,N$}
  \State $W\gets\{o_k\}$
  \State $\mathcal{C}\gets[\,]$

  \For{$l=L,\dots,1$ with $G_l$ unapplied}
    \If{$\supp(G_l)\cap W\neq\emptyset$}
      \State prepend $l$ to $\mathcal{C}$
      \State $W\gets W\cup\supp(G_l)$
    \EndIf
  \EndFor

  \For{$l\in\mathcal{C}$ in increasing order}
    \State inject qumodes in $\supp(G_l)$ not yet in the register
    \State append these injections to $\mathcal{G}$
    \State apply $G_l$ within the register
    \State mark $G_l$ as applied
    \State append $G_l$ to $\mathcal{G}$
  \EndFor

  \If{$o_k$ is not yet in the register}
    \State inject $o_k$
    \State append the injection to $\mathcal{G}$
  \EndIf

  \State emit the memory qumode carrying $o_k$ through a $\pi$ full swap
  \State append the emission to $\mathcal{G}$
\EndFor

\State \Return $\mathcal{G}$

\end{algorithmic}
\end{algorithm}

If $q\in\supp(G_L)$, then $L\in E(q)\subseteq F$. In both $S$ and $\Pi_F$, the gate
$G_L$ forms the row pair of $q$ from the row pairs of the qumodes
$q'\in\supp(G_L)$ immediately before the final gate. For each such $q'$, its
entanglement support in the shortened sequence is contained in
$E(q)\setminus\{L\}$, and therefore in the support-closed set
$F\setminus\{L\}$. By the induction hypothesis, the required row pairs of all such
$q'$ are the same in the two products. Applying the same gate $G_L$ therefore gives
the same row pair for $q$.
\end{proof}

\subsection{Algorithm}
\label{sec:slicing-alg}

Algorithm~\ref{alg:slicing} follows the prescribed emission order. At step $k$, the
reverse scan identifies the complete set of unapplied gates whose effects can
propagate to the output $o_k$. The set $W$ is initialized with $o_k$ and enlarged
whenever an unapplied gate acts on a qumode already in $W$, thereby tracing the
dependence of $o_k$ backward through the gate sequence. The selected gate indices
are stored in $\mathcal{C}$. Although they are identified by a reverse scan, they
are applied in increasing order, so their relative order is exactly the same as in
the original gate sequence. A qumode is injected when it is first required by one
of these gates.

After all gates in $\mathcal{C}$ have been applied, one memory qumode carries the
complete output $o_k$. This qumode is emitted through a full swap. The in-register
operation at step $k$ is therefore the ordered product of the selected gates. Since
these gates are taken directly from the given gate sequence, no additional synthesis
of the in-register operation is required.


\subsection{Correctness}
\label{sec:slicing-correct}

\begin{theorem}[Faithful realization from a gate sequence]
\label{thm:slicing}
Algorithm~\ref{alg:slicing} applies every gate in
Eq.~\eqref{eq:decomposition} exactly once, emits $o_k$ for each
$k$, and faithfully realizes $S$. Each traveling qumode interacts with
the register only once through a full swap, and every in-register
operation is composed of gates from the original gate sequence.
\end{theorem}

\begin{proof}
At every stage, the set of applied gates is support-closed, because it
is enlarged only by adding complete entanglement supports. Marking the
gates after application ensures that no gate is applied more than once.

Before the $k$-th emission, all gates in $E(o_k)$ have been applied.
The reverse scan selects every unapplied gate whose action can influence
$o_k$. Since every qumode is eventually emitted, every gate is
eventually selected and therefore applied exactly once.

At the $k$-th emission, the applied gate set is support-closed and
contains $E(o_k)$. By Lemma~\ref{lem:entsupport}, the row pair of the
memory qumode carrying $o_k$ is therefore identical to the row pair of
$o_k$ in the full transformation $S$. The $\pi$ full swap transfers
this row pair exactly to the vacuum traveling qumode.

No later gate acts on an emitted qumode, since every gate affecting that
qumode is applied before its emission. Previously emitted outputs
therefore remain unchanged. After all $N$ steps, the traveling qumodes
carry all output row pairs of $S$, and the joint transformation is
exactly $S$.
\end{proof}

The schedule is obtained in $O(NL)$ classical time, where $N$ is the
number of output qumodes and $L$ is the number of gates in the input
sequence. No matrix completion, factorization, or synthesis is
required. Each of the original $L$ Gaussian gates is applied exactly
once.

\subsection{Memory Usage and Optimality}
\label{sec:width}

The memory required by Algorithm~\ref{alg:slicing} depends on the chosen gate
sequence. A qumode enters the register when the first applied gate acts on it and
remains there until its emission. Let $r_k$ denote the number of memory qumodes in
the register immediately before the $k$-th emission. The peak occupation is
\begin{equation}
  w=\max_{1\le k\le N} r_k,
  \label{eq:width}
\end{equation}
which defines the entanglement width of the gate sequence for the chosen emission
order.

\begin{proposition}[Entanglement width versus the minimum]
\label{prop:width}
For any gate sequence realizing $S$ and any emission order,
\[
  r_k\ge \tfrac12\rank S[F_k;P_k]
\]
at every step. Consequently, $w\ge\nmem(S)$. Both quantities can be computed in
time polynomial in $N$ and $L$, so the memory overhead of a given gate sequence can
be determined before the protocol is executed.
\end{proposition}

\begin{proof}
By Theorem~\ref{thm:slicing}, the qumode emitted at step $j$ carries the row pair of
$o_j$ in $S$, whose support is $\supp(o_j)$. An input can appear in this row pair
only after that input has been injected. Therefore, by step $k$, the injected set
must contain
$P_k=\bigcup_{j\le k}\supp(o_j)$.

After $k-1$ outputs have been emitted, the register occupation equals the number of
injected qumodes minus $k-1$. Hence,
\[
  r_k\ge |P_k|-(k-1)
  =\tfrac12\rank S[F_k;P_k],
\]
where the equality follows from Corollary~\ref{cor:count}. Taking the maximum over
$k$ gives $w\ge\nmem(S)$.

The values $r_k$ follow directly from the injection and emission times produced by
the $O(NL)$ scheduling procedure, while $\nmem(S)$ is obtained from the counting
formula in Eq.~\eqref{eq:count}.
\end{proof}

The difference $w-\nmem(S)$ measures the additional memory introduced by the chosen
gate sequence. At step $k$, the minimum required occupation is
$|P_k|-(k-1)$, whereas Algorithm~\ref{alg:slicing} retains every input qumode that
has already been involved in an applied gate. Any input used before the first
emission that depends on it therefore contributes to the excess occupation. Such
early injections may be caused by routing operations or by entanglement that is
created and later removed within the circuit.

A gate sequence attains the minimum when no input qumode is introduced before the
first output that depends on it. Whether this condition is satisfied can be
determined by comparing $w$ with $\nmem(S)$ as described in
Proposition~\ref{prop:width}.

\subsection{The General Case}
\label{sec:general}

When a gate sequence is available, Algorithm~\ref{alg:slicing} offers three practical
advantages. First, the schedule can be constructed in $O(NL)$ time without any
linear-algebra operations. Second, the in-register operations are taken directly
from the original circuit, so the protocol uses exactly the $L$ gates already
specified. Third, the gate sequence determines the in-register operation applied
before each emission, removing the freedom left open by the model in
Sec.~\ref{sec:model}.

If no gate sequence for $S$ is known, or if the sequence has a nonzero memory
overhead according to Proposition~\ref{prop:width}, the matrix-based protocol in
Appendix~\ref{app:matrix} can be used instead. It takes only
$S\in\mathrm{Sp}(2N,\mathbb{R})$ and an emission order as input, and realizes $S$
with exactly $\nmem(S)$ memory qumodes for any symplectic transformation. This is
the explicit construction underlying the upper bound in
Theorem~\ref{thm:ncav}.

The matrix-based protocol achieves optimal memory usage at the cost of additional
synthesis. At each emission step, the target output $o_k$ is first expressed in the
basis of the qumodes currently stored in the register. The resulting conjugate pair
is then extended to a full symplectic transformation by symplectic Gram--Schmidt
completion, after which the qumode carrying $o_k$ is emitted through a full swap.
The complete procedure runs in $O(N^4)$ classical time.

The in-register operations are returned as general symplectic matrices $T_k$ acting
on at most $\nmem(S)$ qumodes, rather than as sequences of elementary gates.
Compiling each $T_k$ using the Bloch--Messiah decomposition \cite{Chakhmakhchyan2018} described in
Appendix~\ref{app:compilation} requires $O(\nmem^2)$ beam splitters and
single-qumode squeezers per step, or $O(N\nmem^2)$ elementary gates over the full
protocol. By comparison, Algorithm~\ref{alg:slicing} applies the original $L$ gates
directly.

The matrix $T_k$ is not unique. The protocol fixes only the conjugate pair that is
to be emitted, while the remaining rows can be completed in different symplectic
bases. This freedom affects the resulting gate decomposition, and the
Gram--Schmidt construction generally couples all qumodes currently in the register.
Finding the completion that minimizes the total elementary-gate count when only
$S$ is given remains an open problem.

After introducing the matrix-based protocol, a comparison with the gate-sequence
protocol of Algorithm~1 clarifies their complementary input settings. When an encoding
circuit is available, Algorithm~1 uses its original gate sequence without additional
synthesis, although its memory usage need not be optimal. When only the target
symplectic matrix $S$ is available, the matrix-based protocol achieves the minimum
memory cost but requires the synthesis of new in-register Gaussian operations. The
comparison therefore highlights the tradeoff between reusing an existing circuit and
minimizing the required memory.

Table~\ref{tab:protocols} summarizes the two protocols. For a given gate sequence,
Proposition~\ref{prop:width} determines whether its memory usage matches the minimum
or whether the matrix-based protocol is preferable.

\begin{table}[t]
\caption{Comparison of the gate-sequence protocol in
Algorithm~\ref{alg:slicing} and the matrix-based protocol in
Algorithm~\ref{alg:matrix} of Appendix~\ref{app:matrix}. Here $L$ is the number of
gates in the input sequence, and $\nmem$ is the minimum memory requirement from
Theorem~\ref{thm:ncav}. The gate-sequence protocol uses the original circuit gates
directly, whereas the matrix-based protocol always achieves the minimum memory but
requires synthesized, gauge-dependent in-register operations.}
\label{tab:protocols}
\begin{ruledtabular}
\begin{tabular}{lcc}
 & Gate sequence & Matrix only \\
\colrule
Input & $S=G_L\cdots G_1$ & $S\in\Sp(2N,\R)$ \\
Memory qumodes & $\ge\nmem$ & $\nmem$ \\
Classical runtime & $O(NL)$ & $O(N^4)$ \\
In-register operations & original circuit gates & synthesized \\
Elementary gates & $L$ & $O(N\nmem^2)$ \\
\end{tabular}
\end{ruledtabular}
\end{table}

\section{Emission-order Dependence: The Chain Code}
\label{sec:chain}

Theorem~\ref{thm:ncav} determines the memory cost for any fixed emission order, but
the order itself remains a design choice. This section examines its effect using a
specific encoding circuit. For the same transformation $S$, different emission
orders can change the required memory from a constant to the maximum value $N$.

\subsection{The Chain Code}
\label{sec:chaincode}

Consider a five-qumode one-dimensional encoding circuit consisting of one data qumode
followed by four ancilla qumodes, as shown in Fig.~\ref{fig:chaincode}. The input
$a_1$ carries the data state, while $a_2,\dots,a_5$ are ancilla qumodes prepared in GKP states. Neighboring qumodes are coupled by
two-mode squeezing (TMS), as in GKP-stabilizer codes that encode one oscillator
into several oscillators using TMS interactions with GKP ancilla
qumodes~\cite{Wu2023,Brady2024}.

For the quadrature vector
$(\hat q_j,\hat p_j,\hat q_{j+1},\hat p_{j+1})$, a TMS gate with squeezing
parameter $r$ is represented by
\begin{equation}
  G(r) =
  \begin{pmatrix}
    \cosh r\, I & \sinh r\, Z \\
    \sinh r\, Z & \cosh r\, I
  \end{pmatrix},
  \qquad
  Z =
  \begin{pmatrix}
    1 & 0 \\
    0 & -1
  \end{pmatrix},
  \label{eq:tms}
\end{equation}
where $I$ is the $2\times2$ identity matrix. The chain encoder applies this gate
between consecutive qumodes,
\begin{equation}
  S_{\mathrm{chain}} = G_{N-1}\cdots G_2G_1,
  \label{eq:chain}
\end{equation}
where $G_j$ denotes $G(r)$ acting on qumodes $j$ and $j+1$. Thus, the data qumode
is first coupled to the first ancilla, which is then coupled to the next ancilla,
and so on along the chain. The following analysis depends only on
$S_{\mathrm{chain}}$, not on the particular input states, as established in
Sec.~\ref{sec:model}.

\begin{figure}[t]
  \includegraphics[width=\columnwidth]{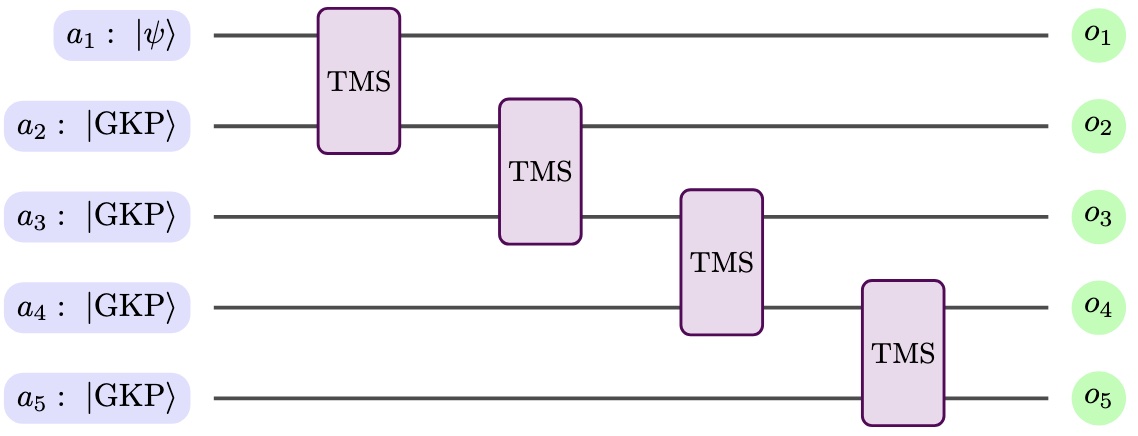}
  \caption{Five-qumode chain encoder. The input qumode $a_1$ carries the data
  state $\ket{\psi}$, while $a_2,\ldots,a_5$ are GKP ancilla qumodes.
  Consecutive qumodes are coupled by nearest-neighbor TMS gates with squeezing
  parameter $r$. The corresponding output qumodes are labeled
  $o_1,\ldots,o_5$.}
  \label{fig:chaincode}
\end{figure}

The support of each output follows from the gate order in
Eq.~\eqref{eq:chain}. Qumode $k$ is acted on only by $G_{k-1}$ and $G_k$, with
$G_0$ and $G_N$ absent. Its row pair becomes final after $G_k$ for
$k\le N-1$, and after $G_{N-1}$ for $k=N$. Since every block of
Eq.~\eqref{eq:tms} is nonzero for $r\neq0$, the dependence on earlier inputs
propagates along the chain:
\begin{equation}
  \supp(o_k) =
  \begin{cases}
    \{a_1,\dots,a_{k+1}\}, & k\le N-1, \\
    \{a_1,\dots,a_N\}, & k=N.
  \end{cases}
  \label{eq:chainsupp}
\end{equation}
Therefore, the final output depends on all input qumodes, while each earlier output
depends only on the inputs up to one position beyond its own.

\subsection{Two Emission Orders}
\label{sec:twoorders}

Two natural choices are the forward order
$o_1,o_2,\dots,o_N$, which follows the chain, and the reverse order
$o_N,o_{N-1},\dots,o_1$. Applying Eq.~\eqref{eq:ncav} to
$S_{\mathrm{chain}}$ shows that these two orders have the largest possible
difference in memory cost. The corresponding protocols for $N=5$ are shown in
Fig.~\ref{fig:chain_emission}.

For the forward order, Eq.~\eqref{eq:chainsupp} gives
$P_k=\{a_1,\dots,a_{k+1}\}$ for $k\le N-1$. Corollary~\ref{cor:count}
therefore gives
$\chi_k=\tfrac12\rank S[F_k;P_k]=|P_k|-(k-1)=2$ for
$k\le N-1$, while $\chi_N=1$. Hence,
$n_{\mathrm{mem}}=2$, independent of $N$. One memory qumode carries the
output $o_k$ being completed, while the other carries the boundary qumode
$k+1$, through which the remaining gates retain access to the injected inputs.

For the reverse order, the first emitted output is $o_N$. Since
$\supp(o_N)=\{a_1,\dots,a_N\}$, all inputs must be injected before the first
emission. Thus, $\chi_1=N$, after which the register occupation decreases as
$\chi_k=N-(k-1)$. The same transformation therefore requires only two memory
qumodes in the forward order but $N$ memory qumodes in the reverse order. For
$N=5$, the two occupation profiles are
$(2,2,2,2,1)$ and $(5,4,3,2,1)$.

Algorithm~\ref{alg:slicing} gives the corresponding gate-level protocols. In the
forward order [Fig.~\ref{fig:chain_emission}(a)], the only unapplied gate needed
to complete $o_k$ is $G_k$. Each step therefore injects one new input qumode,
applies one TMS gate, and emits one output. The protocol reaches a steady pattern
with entanglement width $w=2$, equal to the minimum.

In the reverse order [Fig.~\ref{fig:chain_emission}(b)], the entanglement support
of the first output $o_N$ contains the entire gate sequence. All $N$ inputs and
all $N-1$ TMS gates are therefore required before the first emission. The full
encoding circuit is applied first, followed by the sequential release of all
outputs, giving $w=N$. In both orders, the gate-sequence protocol attains the
minimum for that order. The difference in memory cost is therefore caused entirely
by the emission order.

\begin{figure}[t]
  \includegraphics[width=\columnwidth]{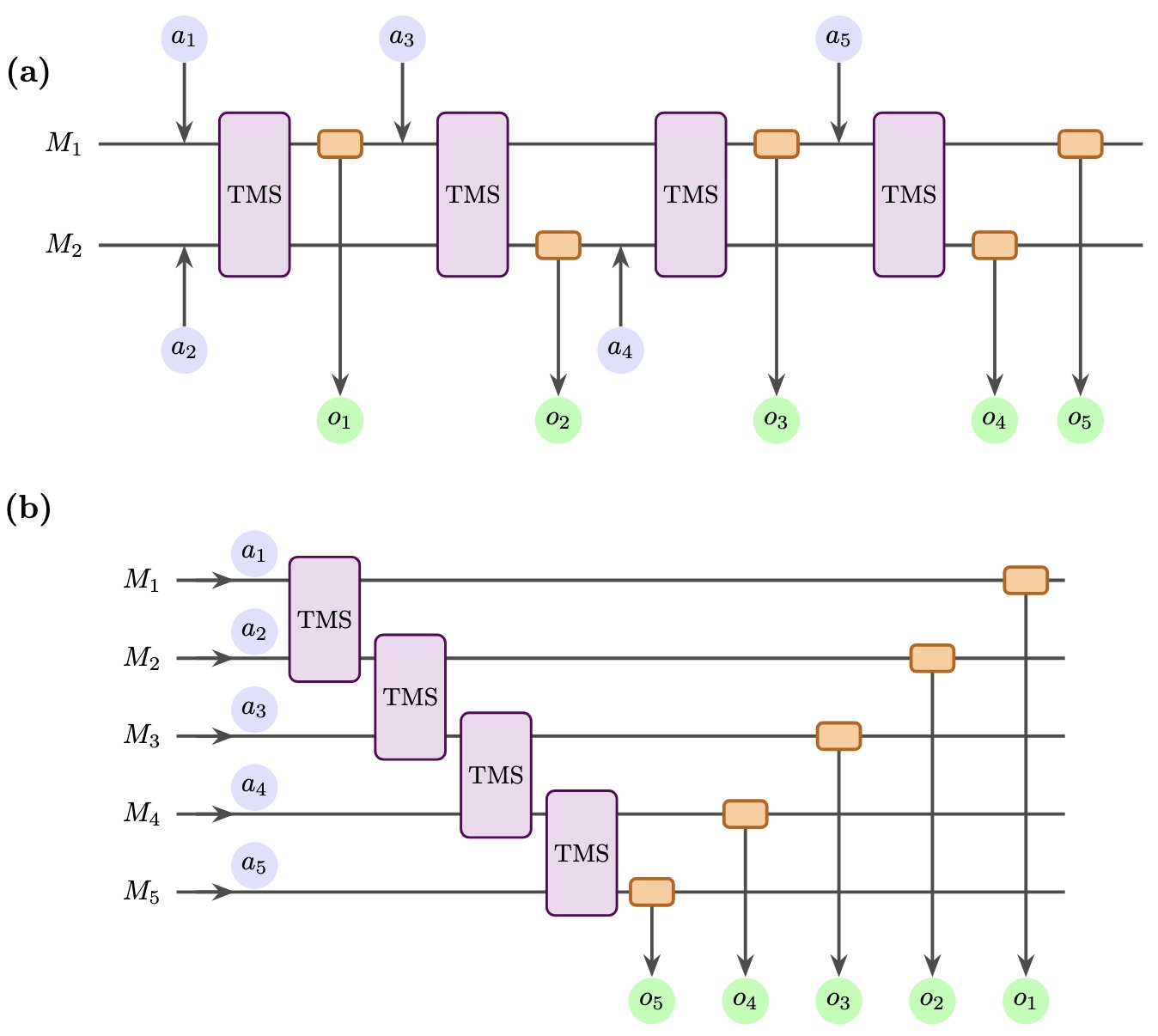}
  \caption{Sequential emission of the $N=5$ chain code in two different orders,
  generated by Algorithm~\ref{alg:slicing}. Horizontal lines denote memory qumodes
  $M_1,M_2,\dots$, vertical arrows denote input injections, and the orange blocks
  denote $\pi$ full swaps to traveling qumodes.
  (a)~Forward order. Two memory qumodes are sufficient, giving
  $n_{\mathrm{mem}}=2$. The first step injects $a_1$ and $a_2$; each subsequent
  step injects the next input into the memory qumode freed by the previous emission,
  applies one TMS gate, and emits the next output. The final step emits $o_5$
  without an additional gate.
  (b)~Reverse order. Since $o_5$ depends on all inputs, all five inputs and all
  four TMS gates are required before the first emission. The encoded outputs are
  then released sequentially, giving $n_{\mathrm{mem}}=5$. Both panels realize
  the same transformation and differ only in the emission order.}
  \label{fig:chain_emission}
\end{figure}

\subsection{The Order-selection Problem}
\label{sec:orderproblem}

The chain code shows that the emission order can change the memory cost of the same
transformation from $O(1)$ to $\Theta(N)$, which is the largest possible variation.
Choosing the emission order is therefore an important part of the protocol design.

Finding the best order is a discrete optimization problem over $N!$ possible
permutations. However, the cost of any fixed order can be evaluated efficiently:
Corollary~\ref{cor:count} reduces it to a linear-time counting procedure. Exhaustive
search is therefore practical only for small $N$. The next section introduces a
heuristic that uses this efficient cost evaluation to select low-memory emission
orders for larger systems.

\section{Heuristic Emission-order Selection}
\label{sec:order}

The previous section introduced the problem of choosing the emission order. This
section presents a deterministic greedy heuristic for an arbitrary
$S\in\mathrm{Sp}(2N,\mathbb{R})$. The heuristic uses only the support sets of the
outputs. Its performance is compared with the certified global optimum on a
running example in this section and with coloring-based alternatives on a larger
benchmark in Appendix~\ref{app:benchmark}.

The output labels $o_1,\dots,o_N$ are fixed by the transformation $S$; only the
order in which these outputs are emitted is varied. This order is recorded by a
permutation $\Xi=(\Xi_1,\dots,\Xi_N)$ of $\{1,\dots,N\}$, where $\Xi_k$ is the
label of the output emitted at step $k$. The sets in Eqs.~\eqref{eq:support}
and~\eqref{eq:Pk} are defined with respect to this order, and
Corollary~\ref{cor:count} gives the memory cost
$n_{\mathrm{mem}}(\Xi)=\max_k\bigl(|P_k|-k+1\bigr)$.

\subsection{Greedy Minimum-frontier Selection}
\label{sec:greedy}

\begin{algorithm}[t]
\caption{Greedy emission-order selection}
\label{alg:greedy}
\begin{algorithmic}[1]

\Statex \textbf{Definitions:}
\Statex $R$: unselected outputs; $J$: required inputs;
$\Xi$: emission order; $c$: maximum register occupation.

\Statex \rule{\linewidth}{0.4pt}

\Statex \textbf{Input:}
\Statex Symplectic matrix $S\in\mathrm{Sp}(2N,\mathbb{R})$.

\Statex \rule{\linewidth}{0.4pt}

\Statex \textbf{Output:}
\Statex Emission order $\Xi$ and memory cost
$c=n_{\mathrm{mem}}(\Xi)$.

\Statex \rule{\linewidth}{0.4pt}

\State compute $\supp(o)$ for every output $o$
\State $R\gets\{o_1,\dots,o_N\}$
\State $J\gets\emptyset$
\State $\Xi\gets[\,]$
\State $c\gets0$

\While{$R\neq\emptyset$}
  \State $o^\ast\gets
  \arg\min_{o\in R}|\supp(o)\setminus J|$
  \Comment{Break ties by output label}
  \State append $o^\ast$ to $\Xi$
  \State $J\gets J\cup\supp(o^\ast)$
  \State $R\gets R\setminus\{o^\ast\}$
  \State $c\gets\max(c,|J|-|\Xi|+1)$
\EndWhile

\State \Return $\Xi,c$

\end{algorithmic}
\end{algorithm}

Algorithm~\ref{alg:greedy} constructs the emission order one output at a time. At
any step, $R$ contains the outputs that have not yet been selected, while $J$
contains the input qumodes already required by the selected outputs. If a remaining
output $o$ is emitted next, the additional input qumodes that must be injected are
those in $\supp(o)\setminus J$.

By Corollary~\ref{cor:count}, minimizing the register occupation at the next step is
therefore equivalent to minimizing $|\supp(o)\setminus J|$. The algorithm selects
the output that requires the fewest new input qumodes, updates $J$, and repeats until
all outputs have been ordered. Ties are resolved by the output labels, making the
procedure deterministic.

After the order has been selected, the protocols in Sec.~\ref{sec:algorithm} and
Appendix~\ref{app:matrix} can be used to realize $S$ with the corresponding memory
cost.

\subsection{Complexity}
\label{sec:greedycost}

Algorithm~\ref{alg:greedy} requires only set operations. If $J$ is stored as a
boolean array, evaluating a candidate output $o$ takes
$O(|\supp(o)|)$ time. At a given step, evaluating all remaining outputs therefore
costs $O\!\left(\sum_{o\in R}|\supp(o)|\right)$. Over all $N$ steps, the total
runtime is
\begin{equation}
  O\!\left(N\sum_o|\supp(o)|\right)
  \le O(N^2s),
  \label{eq:greedy-complexity}
\end{equation}
where $s=\max_o|\supp(o)|$ is the largest output support. No linear-algebra
operations are required.

\subsection{Example: The Analog Shor Code}
\label{sec:shor}

The greedy heuristic is evaluated on the continuous-variable analog of the
nine-qubit Shor code~\cite{Shor1995,Braunstein1998}, whose encoding circuit is shown
in Fig.~\ref{fig:shorcode}. The data state is assigned to qumode $1$, while
qumodes $2,\dots,9$ contain the resource states.

\begin{figure}[t]
  \includegraphics[width=\columnwidth]{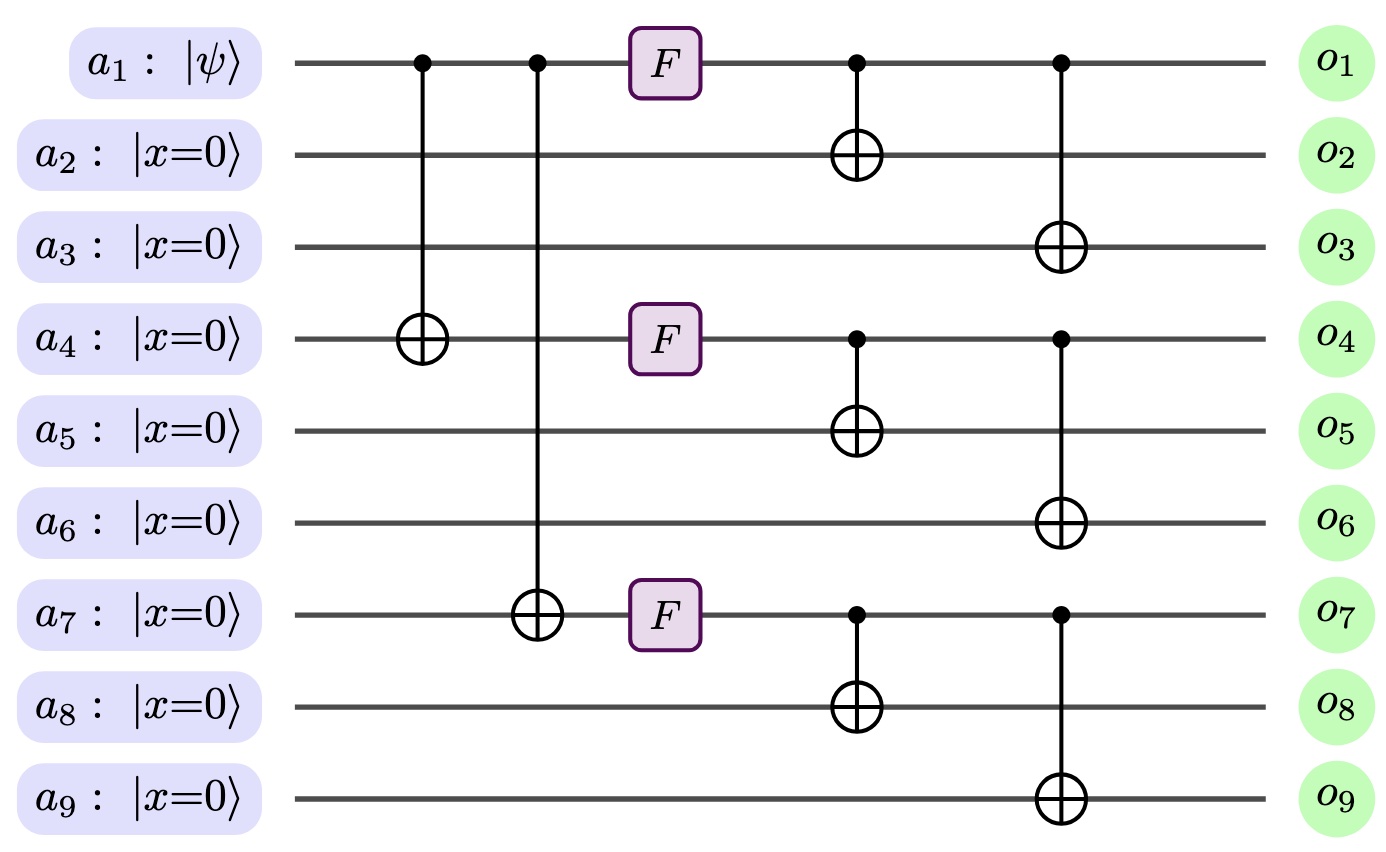}
  \caption{Encoding circuit of the nine-qumode analog Shor code. The input qumode
  $a_1$ carries the data state $\ket{\psi}$, while $a_2,\ldots,a_9$ are initialized
  in the position eigenstate $\ket{x=0}$. The circuit consists of Fourier gates
  $F$ and SUM gates, and the corresponding output qumodes are labeled
  $o_1,\ldots,o_9$.}
  \label{fig:shorcode}
\end{figure}

The memory cost depends strongly on the emission order. Emitting the outputs in the
labeling order $o_1,o_2,\dots,o_9$ requires $n_{\mathrm{mem}}=5$ memory qumodes.
The exact optimum over all emission orders is $n_{\mathrm{mem}}=3$, obtained by a
dynamic program over subsets of emitted outputs with runtime
$O(2^N\mathrm{poly}(N))$. Among 20\,000 uniformly sampled random orders, only
$1.6\%$ achieve cost $3$. The most frequent cost is $5$, and the largest sampled
cost is $7$.

For this encoder, Algorithm~\ref{alg:greedy} reaches the certified optimum and
returns
\begin{equation}
  \Xi=(5,\,6,\,4,\,2,\,1,\,3,\,8,\,7,\,9),
  \label{eq:shororder}
\end{equation}
with occupation profile
$(2,2,2,3,3,2,2,2,1)$ and
$n_{\mathrm{mem}}(\Xi)=3$. The corresponding emission circuit generated by
Algorithm~\ref{alg:slicing} is shown in Fig.~\ref{fig:shor_emission}.

The greedy order follows the hierarchical structure of the encoder and processes
its three-qumode blocks one at a time. During each block, one memory qumode retains
the data qumode, a second retains the central qumode of the block, and a third is
used to complete and emit the other two qumodes in sequence. This structure explains
why three memory qumodes are sufficient.

\begin{figure*}[t]
  \includegraphics[width=0.85\textwidth]{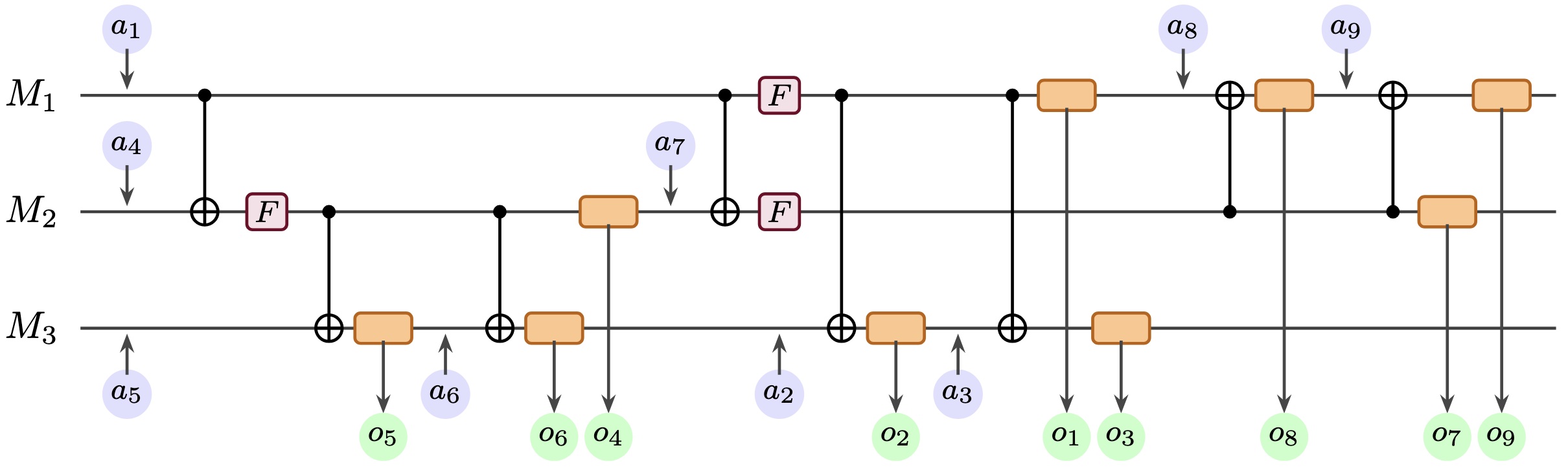}
  \caption{Emission circuit for the analog Shor code in the greedy emission order
  of Eq.~\eqref{eq:shororder}, generated by Algorithm~\ref{alg:slicing}. The
  circuit proceeds from left to right. The horizontal lines represent the three
  memory qumodes $M_1,M_2,M_3$, and the vertical arrows labeled
  $a_1,\ldots,a_9$ indicate input injections. A SUM gate is represented by a
  control dot and a target $\oplus$, while $F$ denotes the Fourier gate. Each
  orange block represents a $\pi$ full swap between a memory qumode and a
  vacuum traveling qumode. The traveling qumode is then emitted as the
  corresponding output $o_k$, indicated by a downward arrow. In the circuit
  shown, the emission operations are distributed among different memory
  qumodes. Nevertheless, only one memory qumode needs to support controlled
  emission into the traveling channel, since additional SWAP gates within the
  register can route each output state to that memory qumode before emission.}
  \label{fig:shor_emission}
\end{figure*}

\subsection{Guarantees and Limitations}
\label{sec:greedylimits}

Algorithm~\ref{alg:greedy} minimizes the register occupation at each individual
step. However, the cost of an emission order is determined by the largest occupation
over the entire sequence, so a choice that is optimal locally may lead to a larger
cost at a later step. The heuristic therefore does not guarantee a globally optimal
order. The computational complexity of finding the exact optimal emission order
also remains open.

The quality of the greedy solution can still be assessed in practice. Its memory
cost is evaluated exactly using Corollary~\ref{cor:count} and can be compared
efficiently with randomly sampled orders. For moderate $N$, the result can also be
compared with the exact optimum obtained by a subset dynamic program based on the
Bellman--Held--Karp approach~\cite{Bellman1962,HeldKarp1962}.

An alternative strategy, inspired by register-allocation
techniques in classical compilers~\cite{Chaitin1982}, is to
order the emissions by graph-coloring heuristics on a
conflict graph derived from the structure of the Gaussian
transformation.
Appendix~\ref{app:benchmark} compares
Algorithm~\ref{alg:greedy} with four such coloring-based
heuristics on a benchmark of 450 random Gaussian transformations
with certified optima. The coloring-based orders recover the
optimum on fewer instances than Algorithm~\ref{alg:greedy},
because they rank the outputs using degree-based proxy
metrics, whereas Algorithm~\ref{alg:greedy} directly tracks
the set of injected inputs $J$, which determines the actual
register occupation.


\section{Memory Scaling for Locally Structured Transformations}
\label{sec:scaling}

The preceding sections showed that the memory cost depends strongly on the emission
order and developed methods for selecting suitable orders for individual encoders.
A broader question is how the minimum memory scales across families of
transformations as the system size increases.

Corollary~\ref{cor:count} shows that the memory cost is determined entirely by the
output supports defined in Eq.~\eqref{eq:support}. For a generic transformation
$S$, every output depends on every input. In that case, all $N$ inputs must be
injected before the first emission, and $N$ memory qumodes are required for any
emission order. A smaller memory register is therefore possible only when the
supports are sparse.

In many physically relevant settings, this sparsity arises from locality. For cluster-state transformations defined on locally connected graphs, each output
depends only on inputs within a local neighborhood
~\cite{Zhang2006,Menicucci2007}. This section shows that, for such transformations, the required memory scales with the size of the boundary between the emitted and unemitted regions rather than with the total number of qumodes. The scaling is independent of the type, number, and
strength of the couplings. It is achieved by a simple sweep through the structure,
without any search over emission orders.

\subsection{Slice-local Transformations}
\label{sec:slicelocal}

Partition the $N$ qumodes into slices $V_1,\dots,V_m$, and let $v(q)$ denote the
index of the slice containing qumode $q$. The transformation $S$ is called
slice-local with range $\rho$ if
\begin{equation}
  |v(a)-v(o)|\le \rho
  \label{eq:slicelocal}
\end{equation}
for every output $o$ and every input $a\in\supp(o)$. Thus, each output depends only
on inputs located within $\rho$ slices of its own slice. The slice order emits all
outputs in $V_1$ first, followed by those in $V_2$, and so on; the order within each
slice is arbitrary.

\begin{theorem}[Memory bound for slice-local transformations]
\label{thm:slab}
Let $S\in\mathrm{Sp}(2N,\mathbb{R})$ be slice-local with range $\rho$ with respect
to the slices $V_1,\dots,V_m$. Under the slice order,
\begin{equation}
  n_{\mathrm{mem}}
  \le
  (\rho+1)\max_i |V_i|.
  \label{eq:slab}
\end{equation}
\end{theorem}

\begin{proof}
Consider the $k$-th emission, and let $r=v(o_k)$ be the slice containing $o_k$.
Under the slice order, all outputs emitted up to this step belong to
$V_1,\dots,V_r$. By Eq.~\eqref{eq:slicelocal}, their supports extend no further
than $\rho$ slices beyond $V_r$. Therefore,
\begin{equation}
  P_k
  =
  \bigcup_{j\le k}\supp(o_j)
  \subseteq
  V_1\cup\cdots\cup V_{\min(r+\rho,m)}.
  \label{eq:Pkslab}
\end{equation}

Before the $k$-th emission, all outputs in $V_1,\dots,V_{r-1}$ have already been
emitted, so
$k-1\ge\sum_{i=1}^{r-1}|V_i|$.
Using Corollary~\ref{cor:count},
$\chi_k=|P_k|-(k-1)
\le\sum_{i=r}^{\min(r+\rho,m)}|V_i|
\le(\rho+1)\max_i|V_i|$.
Taking the maximum over all emission steps proves Eq.~\eqref{eq:slab}.
\end{proof}

\subsection{Graph-local Transformations}
\label{sec:lattices}

Graphs provide a natural class of slice-local transformations. A transformation
$S$ is called graph-local on a graph $G$ if
\begin{equation}
  \supp(o)\subseteq \{o\}\cup\mathcal{N}(o)
  \quad\text{for every output }o,
  \label{eq:graphlocal}
\end{equation}
where $\mathcal{N}(o)$ denotes the set of neighbors of $o$ in $G$. Thus, each
output depends only on its own input and on the inputs associated with adjacent
vertices.

\begin{corollary}[Graph-local transformations]
\label{cor:lattice}
Let $S$ be graph-local on $G$, and partition the vertices into slices
$V_1,\dots,V_m$ such that every edge connects vertices in the same slice or in
adjacent slices. Then $S$ is slice-local with range $\rho=1$, and the slice order
satisfies
\[
n_{\mathrm{mem}}\le 2\max_i|V_i|.
\]
For a $D$-dimensional cubic lattice with linear size $\ell$ and
$N=\ell^D$ qumodes, sliced along one axis,
\begin{equation}
  n_{\mathrm{mem}}
  \le
  2\ell^{D-1}
  =
  2N^{(D-1)/D}.
  \label{eq:lattice}
\end{equation}
\end{corollary}

This condition includes continuous-variable cluster-state transformations
~\cite{Menicucci2006,Menicucci2014}, for which each output depends only on its own
input and those of neighboring vertices. It also allows diagonal, decorated, or
irregular coupling patterns, provided that every coupling connects qumodes in the
same slice or in adjacent slices. The bound depends only on this support pattern,
not on the form or strength of the interactions.

Figure~\ref{fig:front} illustrates the slice order in one and two dimensions. The
emission proceeds slice by slice, while the memory register retains only the
boundary between the emitted and unemitted regions, together with the qumode
currently being completed.

\begin{figure}[t]
  \includegraphics[width=\columnwidth]{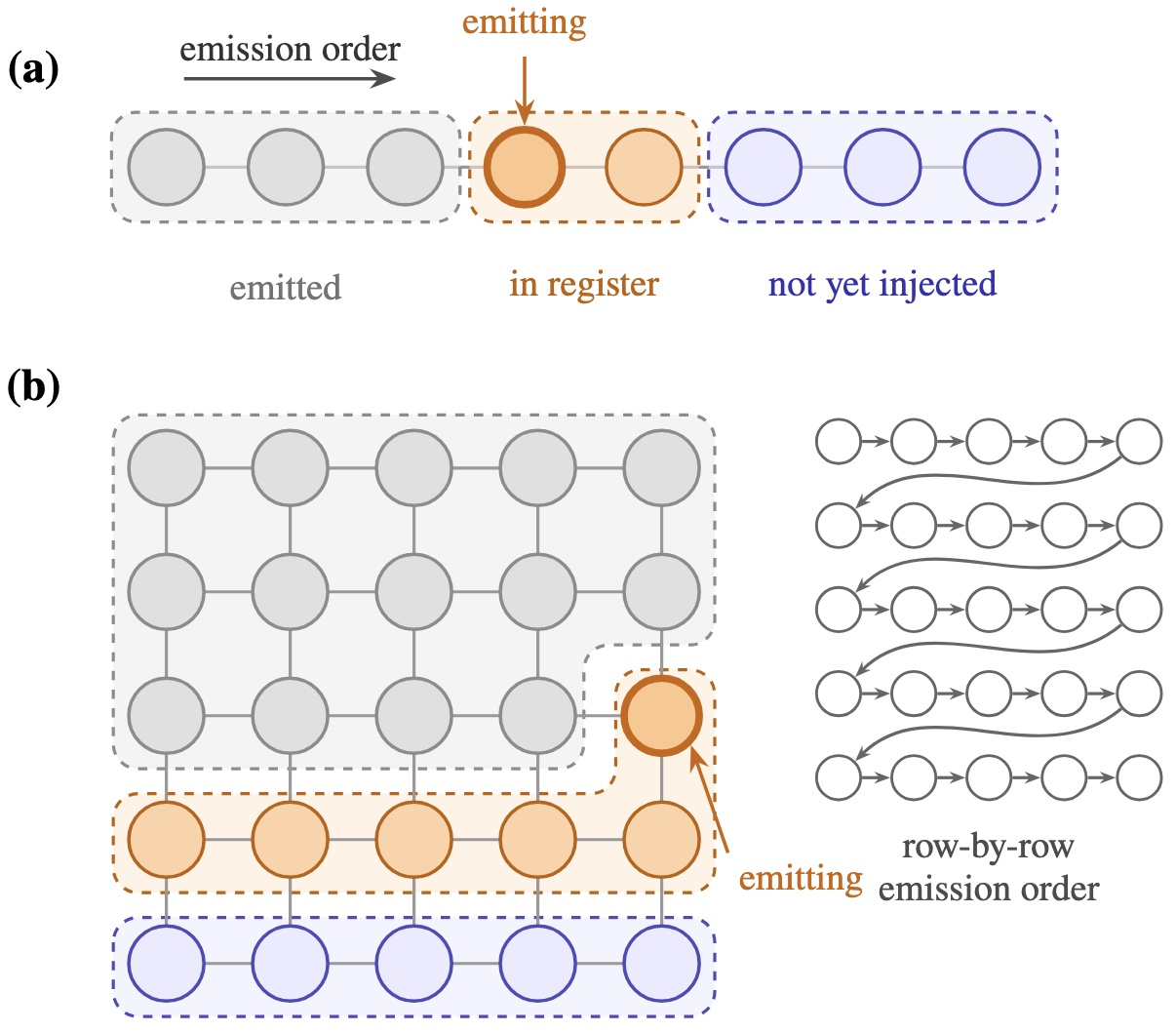}
  \caption{Emission front for graph-local transformations at one step of the slice
  order. The dashed boxes indicate emitted qumodes (gray), qumodes currently held
  in the memory register (orange), and qumodes not yet injected (blue). The qumode
  with the thick outline is the one being emitted.
  (a)~For a one-dimensional chain, each slice contains one qumode, and the outputs
  are emitted sequentially along the chain. At the stage shown, two qumodes are
  held in the memory register.
  (b)~For a two-dimensional lattice of width $\ell$, each row forms one slice. The
  slices are processed row by row, with the emission order illustrated on the
  right. At the stage shown, the memory register contains the qumode being emitted
  together with the injected row ahead of the emission front, for a total of
  $\ell+1$ memory qumodes.}
  \label{fig:front}
\end{figure}

\subsection{Tightness and Significance}
\label{sec:tightness}

For cluster-state transformations defined on locally connected graphs, the
slice order provides a natural sequential implementation whose memory cost is
determined by the corresponding support structure.

\begin{proposition}[Slice-order cost on cubic lattices]
\label{prop:sweep}
Consider the cluster-state transformation on a $D$-dimensional cubic lattice
with nearest-neighbor connectivity and $N=\ell^D$ qumodes, where each lattice
edge represents a controlled-$Z$ interaction. For a slice order taken along one
lattice axis,
\begin{equation}
  n_{\mathrm{mem}}=\ell^{D-1}+1
  \label{eq:sweepcost}
\end{equation}
for every lattice linear size $\ell$ and for any emission order within each
slice.
\end{proposition}

\begin{proof}
Consider the emission of the $c$-th output in slice $V_r$, where
$2\le r<m$. At this step, all outputs in $V_1,\dots,V_{r-1}$ and
$c-1$ outputs in $V_r$ have already been emitted. The injected set contains
$V_1,\dots,V_r$. In particular, since all outputs in $V_{r-1}$ have already
been emitted and collectively depend on all inputs in the adjacent slice
$V_r$, every input in $V_r$ must already have been injected. The injected set
also contains the $c$ inputs in $V_{r+1}$ that are adjacent to the first $c$
outputs emitted from $V_r$.

Hence, the number of injected inputs is $r\ell^{D-1}+c$, while the number of
previously emitted outputs is $(r-1)\ell^{D-1}+c-1$.
Corollary~\ref{cor:count} therefore gives
$\chi_k=\ell^{D-1}+1$, independently of $r$, $c$, and the emission order within
the slice. During the first slice, the occupation increases up to this value,
while during the final slice it decreases below it. This proves
Eq.~\eqref{eq:sweepcost}.
\end{proof}



These results establish the following scaling behavior for locally connected
structures. On a $D$-dimensional cubic lattice with $N=\ell^D$ qumodes, the
memory required by the slice order scales as
$O(\ell^{D-1})=O(N^{(D-1)/D})$. In particular, this scaling is constant in one
dimension and $O(\sqrt{N})$ in two dimensions. More generally, the memory grows
with the size of the boundary between the emitted and unemitted regions rather
than with the total number of qumodes.

The slice order discussed above emits one output at a time. Although this minimizes
the number of simultaneously active emission channels, the total generation time
still grows with the number of outputs. A natural way to shorten this time is to
allow several outputs to be emitted in parallel. Such parallelization can be
implemented in two ways. One approach uses a single memory register equipped with
multiple emission heads, so that several outputs can be released simultaneously
from different memory qumodes. The other uses multiple memory registers, each with
its own emission head, and distributes the generation process among the registers.
Both approaches reduce the number of sequential emission steps. However, multiple
emission heads within a shared register can reuse part of the stored information
and therefore require only a moderate increase in memory, whereas multiple
independent registers duplicate part of the memory associated with the emission
front. Consequently, when simultaneous release from several memory qumodes is
experimentally available, parallel emission within a single register provides a
more memory-efficient tradeoff between generation time and hardware resources.

A similar slice-by-slice generation structure appears in time-multiplexed optical
experiments, where two-dimensional cluster states are generated by retaining
qumodes in delay lines while earlier outputs are emitted
~\cite{Asavanant2019,Larsen2019}. These experiments prepare specific cluster states
from fixed squeezed-state inputs. In contrast, Theorem~\ref{thm:slab} applies to
the full Gaussian transformation and therefore to arbitrary input states. It also
applies to any coupling pattern satisfying the locality condition, while
Eq.~\eqref{eq:ncav} gives the exact memory cost for a specified emission order.


\section{Conclusion and Outlook}
\label{sec:conclusion}

This work establishes the optimal memory cost of sequentially realizing a multimode
Gaussian transformation. For any $S\in\mathrm{Sp}(2N,\mathbb{R})$ and any fixed
emission order, the minimum number of memory qumodes is given by
Eq.~\eqref{eq:ncav}. Corollary~\ref{cor:count} reduces this rank expression to the
number of input qumodes injected so far minus the number of outputs already emitted.
The memory cost therefore depends only on the support pattern of $S$ and can be
evaluated in time linear in the size of the support data.

The minimum memory cost for an arbitrary symplectic matrix is achieved by the
matrix-based protocol in Appendix~\ref{app:matrix}. When a gate sequence is
available, Algorithm~\ref{alg:slicing} instead constructs the sequential
implementation directly from the original gates. Its memory usage need not attain
the minimum, because the given gate decomposition may not be the most
memory-efficient realization of the same transformation. The emission order is
also an important design choice, since different orders can require substantially
different numbers of memory qumodes for the same gate sequence. To address this
ordering problem, Sec.~\ref{sec:order} introduces a greedy algorithm that selects
the output requiring the fewest new input qumodes at each step and generally
produces memory-efficient emission orders. For transformations with local support
on a $D$-dimensional cubic lattice of $N$ qumodes, the memory grows only with the size of
the boundary between the emitted and unemitted regions, giving
$O(N^{(D-1)/D})$ scaling.

Although the implemented transformations are Gaussian, the protocols do not
assume Gaussian input states and therefore also apply to non-Gaussian inputs,
including GKP and cat states. Since such states provide the non-Gaussian resources
required for universal continuous-variable quantum computation, the proposed
framework offers a resource-efficient and explicit scheme for communication
between modules in modular architectures for universal quantum computation.

Several questions remain open. For qubit graph states, optimizing the photon
emission order to minimize the number of emitters is NP-hard~\cite{Li2022}.
Whether or not optimizing the emission order of an arbitrary Gaussian transformation to
minimize the required number of memory qumodes is also NP-hard remains unknown.
For locally connected lattice families, characterizing how the minimum memory cost
depends on the lattice geometry and emission order also remains an open problem. In
the matrix-based protocol, the symplectic completion used to construct the
in-register operations is generally not unique, and the resulting symplectic
matrices can admit many gate decompositions. Selecting the completions and
decompositions that minimize the total number of beam splitters and
single-qumode squeezers is another open problem. Finally, the present analysis assumes ideal operations. Extending it to include memory decay, imperfect full swaps, and other hardware errors is necessary for a complete assessment of sequential Gaussian implementations.

\begin{acknowledgments}
This work was supported by the U.S. Department of Energy, Office of
Science, Advanced Scientific Computing Research, under Contract
No.~DE-SC0025384. This work was also funded in part by NSF grants
No.~OMA-2120757, No.~PHY-2325080, No.~OSI-2410675, and No.~OSI-2531350.
\end{acknowledgments}

\appendix

\section{Sequential Generation from the Symplectic Matrix}
\label{app:matrix}

This appendix considers the case in which the target transformation is specified
only by a symplectic matrix $S\in\mathrm{Sp}(2N,\mathbb{R})$, with no gate
sequence available. An explicit algorithm is given that realizes $S$ faithfully
under sequential emission by single full swaps while using exactly
$n_{\mathrm{mem}}(S)$ memory qumodes, for any symplectic $S$ and any emission
order. This construction provides the upper bound in
Theorem~\ref{thm:ncav} and can be used whenever no gate sequence is known or when
the entanglement width of a known sequence exceeds $n_{\mathrm{mem}}(S)$.

Unlike Algorithm~\ref{alg:slicing}, the in-register operations cannot be taken
directly from a circuit and must instead be constructed from $S$. At each step, the
next output is expressed in the basis of the qumodes currently stored in the
register. This expression is then extended to a full symplectic operation on the
register, after which the corresponding qumode is emitted. All steps are obtained
from explicit linear-algebra operations.

\subsection{Heisenberg Representation}
\label{app:heisenberg}

The algorithm is formulated in the Heisenberg picture. If the register contains
$m$ memory qumodes, their quadratures are collected in a matrix
$M\in\mathbb{R}^{2m\times 2N}$. Each qumode contributes a
$2\times 2N$ row pair expressing its two quadratures in terms of the
$2N$ input quadratures.

Because the stored qumodes satisfy the canonical commutation relations,
\begin{equation}
  M\Omega M^{\top}=\Omega_m,
  \qquad
  \Omega_m=
  \bigoplus_{j=1}^{m}
  \begin{pmatrix}
    0 & 1\\
    -1 & 0
  \end{pmatrix}.
  \label{eq:frame}
\end{equation}
Injecting an input qumode appends its identity row pair to $M$. An in-register
Gaussian operation acts as $M\mapsto TM$, with $T$ symplectic with respect to
$\Omega_m$, and therefore preserves Eq.~\eqref{eq:frame}. Emission removes one
conjugate row pair, leaving a symplectic frame for the remaining $m-1$ qumodes.

The construction uses two earlier results: Different output qumodes commute by
Eq.~\eqref{eq:symframe}, and all inputs required for $o_k$ have been injected by
step $k$ because $\supp(o_k)\subseteq P_k$.

\subsection{Algorithm}
\label{app:algorithm}

\begin{algorithm}[t]
\caption{Sequential generation of $S$ from its symplectic matrix}
\label{alg:matrix}
\begin{algorithmic}[1]

\Statex \textbf{Definitions:}
\Statex $\mathrm{Mem}$: row-pair representations of the current memory qumodes.
\Statex $\mathrm{Injected}$: input qumodes already injected.
\Statex $M=\mathrm{stack}(\mathrm{Mem})$: register matrix.
\Statex $e_a$: identity row pair of input qumode $a$.
\Statex $\alpha=S[o_k,:]$: row pair of output qumode $o_k$.
\Statex $\mathcal{G}$: constructed operation sequence.

\Statex \rule{\linewidth}{0.4pt}

\Statex \textbf{Input:}
\Statex Symplectic matrix $S\in\Sp(2N,\mathbb{R})$.
\Statex Emission order $o_1,\dots,o_N$.

\Statex \rule{\linewidth}{0.4pt}

\Statex \textbf{Output:}
\Statex Operation sequence $\mathcal{G}$ emitting outputs in the order
$o_1,\dots,o_N$.

\Statex \rule{\linewidth}{0.4pt}

\State $\mathcal{G}\gets[\,]$
\State $\mathrm{Mem}\gets[\,]$
\State $\mathrm{Injected}\gets\emptyset$

\For{$k=1,\dots,N$}
  \For{$a\in\supp(o_k)\setminus\mathrm{Injected}$}
    \State append $e_a$ to $\mathrm{Mem}$
    \State $\mathrm{Injected}\gets\mathrm{Injected}\cup\{a\}$
    \State append $\mathrm{inject}(a)$ to $\mathcal{G}$
  \EndFor

  \State $m\gets|\mathrm{Mem}|$
  \State $M\gets\mathrm{stack}(\mathrm{Mem})$
  \State $\alpha\gets S[o_k,:]$
  \State $A\gets\alpha\Om M^{\top}\Om_m^{-1}$
  \Comment{$\alpha=AM$}
  \State $T\gets\mathrm{Complete}(A,\Om_m)$
  \Comment{$T[1{:}2]=A$}
  \State $M\gets TM$
  \State append $T$ to $\mathcal{G}$
  \State $f_k\gets M[1{:}2,:]$
  \State emit $f_k$ through a $\pi$ full swap
  \State append the emission to $\mathcal{G}$
  \State $\mathrm{Mem}\gets$ row pairs of $M$ after removing the first qumode
\EndFor

\State \Return $\mathcal{G}$

\end{algorithmic}
\end{algorithm}

At step $k$, Algorithm~\ref{alg:matrix} first injects any input qumodes in
$\supp(o_k)$ that are not yet present in the register. It then constructs an
in-register symplectic operation that places the output qumode $o_k$ on the first
memory qumode.

Let $\alpha=S[o_k,:]$ denote the row pair of $o_k$ in the input quadratures. Since
the current memory qumodes are represented by the rows of $M$, a matrix
$A\in\mathbb{R}^{2\times2m}$ satisfying $\alpha=AM$ expresses $o_k$ in the
coordinates of the register. Lemma~\ref{lem:sep} shows that this matrix is unique
and is given by
$A=\alpha\Om M^{\top}\Om_m^{-1}$. The same lemma shows that the two rows of $A$
form a normalized conjugate pair with respect to $\Om_m$.

The subroutine $\mathrm{Complete}(A,\Om_m)$ extends this pair to a full symplectic
matrix. Starting from the pair $(q,p)$ given by the two rows of $A$, each additional
vector $v$ is made symplectically orthogonal to the previously selected pairs using
\begin{equation}
  v\mapsto
  v-\{v,p\}_{\Om_m}q+\{v,q\}_{\Om_m}p.
  \label{eq:gramschmidt}
\end{equation}
The remaining vectors are then paired and normalized until a matrix
$T\in\mathbb{R}^{2m\times2m}$ is obtained with
$T\Om_mT^{\top}=\Om_m$ and $T[1{:}2]=A$.

After all gates in $\mathcal{C}$ have been applied, one memory
qumode carries the complete output $o_k$. This qumode is then
emitted through a full swap. The in-register operation at step $k$
is therefore obtained by taking the selected gates directly from
the original circuit while preserving their original relative order.
Equivalently, it is the ordered subsequence of the original gate
sequence specified by $\mathcal{C}$. Hence, neither the gate content
nor the gate ordering needs to be resynthesized.

\subsection{Correctness}
\label{app:correctness}

\begin{lemma}[Separability]
\label{lem:sep}
At step $k$, the output row pair $\alpha=S[o_k,:]$ lies in the row space of $M$.
Therefore, there is a unique matrix $A$ such that $\alpha=AM$. Moreover,
$A\Om_mA^{\top}=
\bigl(\begin{smallmatrix}0&1\\-1&0\end{smallmatrix}\bigr)$.
\end{lemma}

\begin{proof}
Since $\supp(o_k)\subseteq P_k$ by
Eqs.~\eqref{eq:support} and~\eqref{eq:Pk}, the row pair
$\alpha=S[o_k,:]$ depends only on inputs that have already been injected.

The phase space of the injected inputs is spanned by the previously emitted output
qumodes and the qumodes remaining in the register. Each emission removes one output
row pair, while the remaining register qumodes commute with that output. By
Eq.~\eqref{eq:symframe}, $o_k$ also commutes with every previously emitted output.
Hence, $\alpha$ has no component along the removed output pairs and must lie in the
row space of $M$.

Because $M$ has full row rank, the solution to $\alpha=AM$ is unique. Using
Eq.~\eqref{eq:frame},
$A=\alpha\Om M^{\top}\Om_m^{-1}$. Finally,
$A\Om_mA^{\top}=(AM)\Om(AM)^{\top}
=\alpha\Om\alpha^{\top}
=\bigl(\begin{smallmatrix}0&1\\-1&0\end{smallmatrix}\bigr)$
by Eq.~\eqref{eq:symframe}. Thus, the rows of $A$ form a normalized conjugate pair,
and the symplectic completion is well defined.
\end{proof}

\begin{lemma}[Clean emission]
\label{lem:clean}
After the update $M\gets TM$, the first memory qumode carries the output $o_k$.
Emitting this qumode through a full swap preserves the information already stored
in the register for all later outputs.
\end{lemma}

\begin{proof}
By construction, $T[1{:}2]=A$. Therefore, the first row pair of $TM$ is
$AM=\alpha=S[o_k,:]$, so the first memory qumode carries the complete output
$o_k$.

Since $T$ is invertible, the update $M\mapsto TM$ changes only the symplectic basis
of the register and leaves its row space unchanged. Consider a later output
$o_{k'}$ with $k'>k$. By Eq.~\eqref{eq:symframe}, $o_{k'}$ is symplectically
orthogonal to $o_k$. The part of $o_{k'}$ supported on inputs not yet injected is
also symplectically orthogonal to $o_k$, because $o_k$ is supported only on inputs
already present in the register. It follows that the component of $o_{k'}$ currently
stored in the register is symplectically orthogonal to the first row pair of $TM$.

The remaining rows of $TM$ span the symplectic complement of this first row pair
within the register. Removing the first memory qumode therefore preserves all
information currently stored for the later emissions. The full swap transfers
$o_k$ to the vacuum traveling qumode and completes the emission.
\end{proof}

\begin{theorem}[Faithful realization at the minimum]
\label{thm:matrix}
Algorithm~\ref{alg:matrix} emits $f_k=o_k$ at every step and therefore realizes
$S$ faithfully. Each traveling qumode interacts with the register only once through
a full swap. The peak register occupation is
$n_{\mathrm{mem}}(S)=\max_k\tfrac12\rank S[F_k;P_k]$.
\end{theorem}

\begin{proof}
Lemmas~\ref{lem:sep} and~\ref{lem:clean} imply that the $k$-th traveling qumode
carries the output row pair $o_k$. After all $N$ emissions, the traveling qumodes
therefore carry all output row pairs of $S$, so the resulting transformation is
exactly $S$.

The in-register symplectic operations do not change the number of qumodes. The
register size increases only when an input is injected and decreases only when an
output is emitted. Immediately before the $k$-th emission, the injected set is
exactly $P_k$, and $k-1$ outputs have already been emitted. The register therefore
contains $|P_k|-(k-1)$ memory qumodes. By Corollary~\ref{cor:count}, this equals
$\tfrac12\rank S[F_k;P_k]$. Taking the maximum over $k$ gives
$n_{\mathrm{mem}}(S)$, so Algorithm~\ref{alg:matrix} achieves the minimum.
\end{proof}

\subsection{Complexity and Compilation}
\label{app:compilation}

At step $k$, the matrix
$A=\alpha\Om M^{\top}\Om_m^{-1}$ is computed in $O(Nm)$ time. The
symplectic completion and the update $M\gets TM$ require $O(m^3)$ and
$O(m^2N)$ time, respectively. Since
$m\le\nmem(S)\le N$, each step costs at most $O(N^3)$, and the full
algorithm runs in $O(N^4)$ time.

Algorithm~\ref{alg:matrix} returns each in-register operation as a
symplectic matrix $T_k$ acting on at most $\nmem(S)$ memory qumodes.
This matches the model of Sec.~\ref{sec:model}, where arbitrary Gaussian
operations within the register are assumed to be available. To implement
$T_k$ using elementary gates, two choices must be considered: the
symplectic completion used to construct $T_k$, and the decomposition of
the resulting matrix into hardware operations.

The protocol fixes only the first row pair of $T_k$,
$T_k[1{:}2]=A$, which places the output $o_k$ on the first memory
qumode. The remaining rows may be chosen in any way that makes $T_k$
symplectic. Different choices produce the same emitted output and use
the same number of memory qumodes, but may lead to different elementary
gate counts. The subroutine $\mathrm{Complete}$ selects one such
completion using symplectic Gram--Schmidt orthogonalization against the
standard basis. This construction generally couples all qumodes
currently held in the register. When a gate sequence is available,
Algorithm~\ref{alg:slicing} avoids this choice by taking the
in-register operations directly from the given sequence.

For a fixed $T_k$, the Bloch--Messiah decomposition gives
\begin{equation}
  T_k
  =
  O_{\mathrm{out}}
  \left[
    \bigoplus_j \mathrm{Sq}(r_j)
  \right]
  O_{\mathrm{in}},
  \label{eq:bm}
\end{equation}
where $O_{\mathrm{in}}$ and $O_{\mathrm{out}}$ are passive
interferometers and $\mathrm{Sq}(r_j)$ denotes single-qumode squeezing.
The passive interferometers can be decomposed into beam-splitter
networks using either the triangular~\cite{Reck1994} or
rectangular~\cite{Clements2016} construction. Computing the
Bloch--Messiah decomposition costs $O(m^3)$ time. The resulting circuit
contains $m(m-1)$ beam splitters and $m$ single-qumode squeezers, giving
a total elementary-gate count of $O(N\nmem^2)$ over all emission steps.

The elementary-gate count depends on the chosen symplectic completion.
Further reductions may also be possible by commuting, merging, or
canceling gates across successive steps, although methods developed
for qubit circuits~\cite{Maslov2008,Nam2018} do not directly apply to
beam splitters and single-qumode squeezers. Determining the minimum
elementary-gate count from $S$ alone remains open. This freedom affects
only the compilation cost, while the minimum number of memory qumodes remains
fixed by Eq.~\eqref{eq:ncav}.

\section{Benchmark Comparison with Coloring-Based Emission Orders}
\label{app:benchmark}

This appendix compares Algorithm~\ref{alg:greedy} with emission-order
heuristics inspired by graph-coloring-based register
allocation~\cite{Chaitin1982}. The comparison uses a benchmark of 450
random Gaussian transformations on $N=12$ qumodes, generated from a fixed seed
and divided into three families of 150 instances each. In the sparse
random family, the number of gates is drawn uniformly from the integers
in $[N,2N]$, and each gate acts on a pair of distinct qumodes drawn
uniformly at random from all pairs, with repeated pairs permitted. In
the random tree family, the qumodes are placed in a uniformly random
order, and each qumode after the first is coupled to one uniformly
chosen earlier qumode, giving $N-1$ gates that form a tree. The
application order of these gates is then randomized. In the brickwork
family, the qumodes are arranged on a line and the gates form
alternating layers of nearest-neighbor pairs, shifted by one qumode
between consecutive layers. The number of layers is one, two, or three
with probabilities $0.1$, $0.45$, and $0.45$, respectively, and each
potential gate is included independently with probability $0.8$.
Every gate is a generic two-mode Gaussian gate, and the output supports
are obtained by propagating the supports through the gate list.

For every instance, the exact minimum memory cost over all emission
orders is certified by the Bellman--Held--Karp dynamic
program~\cite{Bellman1962,HeldKarp1962}. The system size $N=12$ is
chosen so that this certification remains feasible for all 450
instances.

The coloring-based heuristics operate on a conflict graph whose vertices
are the outputs. The next output to emit is the remaining vertex of
minimum degree, or minimum weighted degree in the weighted variants.
The degrees are recomputed after each selection, and ties are broken by
output label. Two graph constructions are considered. In the
support-overlap graph, two outputs are adjacent when their supports
intersect, and the weighted variant assigns the edge between $o_i$ and
$o_j$ the weight
$|\mathrm{supp}(o_i)\cap\mathrm{supp}(o_j)|$. In the gate graph,
outputs $o_i$ and $o_j$ are adjacent when at least one gate in the
circuit couples qumodes $i$ and $j$, and the weighted variant assigns
to each edge a weight equal to the number of such gates.

Table~\ref{tab:coloring-benchmark} reports, for each heuristic, the
number of instances on which the resulting memory cost exceeds the
certified optimum by at most a given number of qumodes.
Algorithm~\ref{alg:greedy} attains the optimum on 313 of the 450
instances, uses at most one additional memory qumode on 424 instances,
and never exceeds the optimum by more than three. The best-performing
coloring variant, the weighted support-overlap graph, attains the
optimum on 127 instances and exceeds it by up to four. The unweighted
gate graph attains the optimum on 150 instances, more than either
support-overlap variant, but its excess reaches seven qumodes, and 36 of
its orders exceed the optimum by five or more. The advantage of
Algorithm~\ref{alg:greedy} persists within each family, i.e., it is optimal on
108 of 150 sparse instances, 57 of 150 tree instances, and 148 of 150
brickwork instances, in each case more than any coloring variant.

The comparison indicates why the greedy rule outperforms the
coloring-based orders. The degree of a vertex in a conflict graph is
only an indirect proxy for the memory pressure associated with an
output. By Corollary~\ref{cor:count}, the register occupation is
determined exactly by the number of injected inputs, and
Algorithm~\ref{alg:greedy} minimizes this quantity directly at every
step. The coloring-based orders optimize the proxy rather than the
occupation itself, which leads to heavier tails in the excess
distribution.

\begin{table}[t]
\caption{Cumulative excess over the certified optimum on the
450-instance benchmark. The column labeled $x$ gives the number of
instances on which the heuristic uses at most $x$ memory qumodes more
than the exact minimum. The last column gives the largest excess.}
\label{tab:coloring-benchmark}
\begin{ruledtabular}
\begin{tabular}{lcccccc}
 & \multicolumn{5}{c}{Excess $\le x$} & \\
\cline{2-6}
Method & $0$ & $1$ & $2$ & $3$ & $4$ & Worst \\
\colrule
Algorithm~\ref{alg:greedy}       & 313 & 424 & 447 & 450 & 450 & $+3$ \\
Support overlap, weighted        & 127 & 281 & 416 & 447 & 450 & $+4$ \\
Support overlap, unweighted      &  70 & 172 & 285 & 390 & 434 & $+7$ \\
Gate graph, unweighted           & 150 & 196 & 267 & 350 & 414 & $+7$ \\
Gate graph, weighted             & 131 & 176 & 266 & 350 & 415 & $+7$ \\
\end{tabular}
\end{ruledtabular}
\end{table}

\bibliographystyle{apsrev4-2}
\bibliography{References}

@article{Kimble2008,
  title = {The quantum internet},
  author = {Kimble, H. J.},
  journal = {Nature},
  volume = {453},
  pages = {1023--1030},
  year = {2008},
  doi = {10.1038/nature07127}
}

@article{4rf7-9tfx,
  title = {Hybrid Oscillator-Qubit Quantum Processors: Instruction Set Architectures, Abstract Machine Models, and Applications},
  author = {Liu, Yuan and Singh, Shraddha and Smith, Kevin C. and Crane, Eleanor and Martyn, John M. and Eickbusch, Alec and Schuckert, Alexander and Li, Richard D. and Sinanan-Singh, Jasmine and Soley, Micheline B. and Tsunoda, Takahiro and Chuang, Isaac L. and Wiebe, Nathan and Girvin, Steven M.},
  journal = {PRX Quantum},
  volume = {7},
  issue = {1},
  pages = {010201},
  numpages = {166},
  year = {2026},
  month = {Jan},
  publisher = {American Physical Society},
  doi = {10.1103/4rf7-9tfx},
  url = {https://link.aps.org/doi/10.1103/4rf7-9tfx}
}

@article{Wehner2018,
  title = {Quantum internet: A vision for the road ahead},
  author = {Wehner, Stephanie and Elkouss, David and Hanson, Ronald},
  journal = {Science},
  volume = {362},
  pages = {eaam9288},
  year = {2018},
  doi = {10.1126/science.aam9288}
}

@article{Monroe2014,
  title = {Large-scale modular quantum-computer architecture with atomic memory and photonic interconnects},
  author = {Monroe, C. and Raussendorf, R. and Ruthven, A. and Brown, K. R. and Maunz, P. and Duan, L.-M. and Kim, J.},
  journal = {Phys. Rev. A},
  volume = {89},
  pages = {022317},
  year = {2014},
  doi = {10.1103/PhysRevA.89.022317}
}

@article{Nickerson2014,
  title = {Freely scalable quantum technologies using cells of 5-to-50 qubits with very lossy and noisy photonic links},
  author = {Nickerson, Naomi H. and Fitzsimons, Joseph F. and Benjamin, Simon C.},
  journal = {Phys. Rev. X},
  volume = {4},
  pages = {041041},
  year = {2014},
  doi = {10.1103/PhysRevX.4.041041}
}

@article{Reagor2016,
  title = {Quantum memory with millisecond coherence in circuit QED},
  author = {Reagor, Matthew and Pfaff, Wolfgang and Axline, Christopher
            and Heeres, Reinier W. and Ofek, Nissim and Sliwa, Katrina
            and Holland, Eric and Wang, Chen and Blumoff, Jacob
            and Chou, Kevin and Hatridge, Michael J. and Frunzio, Luigi
            and Devoret, M. H. and Jiang, Liang and Schoelkopf, R. J.},
  journal = {Phys. Rev. B},
  volume = {94},
  pages = {014506},
  year = {2016},
  doi = {10.1103/PhysRevB.94.014506}
}

@article{Pfaff2017,
  title = {Controlled release of multiphoton quantum states from a microwave cavity memory},
  author = {Pfaff, W. and Axline, C. J. and Burkhart, L. D. and Vool, U. and Reinhold, P. and Frunzio, L. and Jiang, L. and Devoret, M. H. and Schoelkopf, R. J.},
  journal = {Nat. Phys.},
  volume = {13},
  pages = {882--887},
  year = {2017},
  doi = {10.1038/nphys4143}
}

@article{Axline2018,
  title = {On-demand quantum state transfer and entanglement between remote microwave cavity memories},
  author = {Axline, C. J. and Burkhart, L. D. and Pfaff, W. and Zhang, M. and Chou, K. and Campagne-Ibarcq, P. and Reinhold, P. and Frunzio, L. and Girvin, S. M. and Jiang, L. and Devoret, M. H. and Schoelkopf, R. J.},
  journal = {Nat. Phys.},
  volume = {14},
  pages = {705--710},
  year = {2018},
  doi = {10.1038/s41567-018-0115-y}
}

@article{BraunsteinVanLoock2005,
  title = {Quantum information with continuous variables},
  author = {Braunstein, Samuel L. and van Loock, Peter},
  journal = {Rev. Mod. Phys.},
  volume = {77},
  pages = {513--577},
  year = {2005},
  doi = {10.1103/RevModPhys.77.513}
}

@article{Weedbrook2012,
  title = {Gaussian quantum information},
  author = {Weedbrook, Christian and Pirandola, Stefano and Garc{\'i}a-Patr{\'o}n, Ra{\'u}l and Cerf, Nicolas J. and Ralph, Timothy C. and Shapiro, Jeffrey H. and Lloyd, Seth},
  journal = {Rev. Mod. Phys.},
  volume = {84},
  pages = {621--669},
  year = {2012},
  doi = {10.1103/RevModPhys.84.621}
}

@article{Braunstein1998,
  title = {Error correction for continuous quantum variables},
  author = {Braunstein, Samuel L.},
  journal = {Phys. Rev. Lett.},
  volume = {80},
  pages = {4084--4087},
  year = {1998},
  doi = {10.1103/PhysRevLett.80.4084}
}

@article{LloydSlotine1998,
  title = {Analog quantum error correction},
  author = {Lloyd, Seth and Slotine, Jean-Jacques E.},
  journal = {Phys. Rev. Lett.},
  volume = {80},
  pages = {4088--4091},
  year = {1998},
  doi = {10.1103/PhysRevLett.80.4088}
}

@article{Guo2026,
  title = {Concatenated dual displacement code for continuous-variable quantum error correction},
  author = {Guo, Fucheng and Mueller, Frank and Liu, Yuan},
  journal = {Phys. Rev. Research},
  volume = {8},
  pages = {023158},
  year = {2026},
  doi = {10.1103/d1zs-cy4t}
}

@article{Gao2018,
  title = {Programmable interference between two microwave quantum memories},
  author = {Gao, Yvonne Y. and Lester, Brian J. and Zhang, Yaxing and Wang, Chen and Rosenblum, Serge and Frunzio, Luigi and Jiang, Liang and Girvin, S. M. and Schoelkopf, R. J.},
  journal = {Phys. Rev. X},
  volume = {8},
  pages = {021073},
  year = {2018},
  doi = {10.1103/PhysRevX.8.021073}
}

@article{Chapman2023,
  title = {High-on-off-ratio beam-splitter interaction for gates on bosonically encoded qubits},
  author = {Chapman, Benjamin J. and de Graaf, Stijn J. and Xue, Sophia H. and Zhang, Yaxing and Teoh, James and Curtis, Jacob C. and Tsunoda, Takahiro and Eickbusch, Alec and Read, Alexander P. and Koottandavida, Akshay and Mundhada, Shantanu O. and Frunzio, Luigi and Devoret, M. H. and Girvin, S. M. and Schoelkopf, R. J.},
  journal = {PRX Quantum},
  volume = {4},
  pages = {020355},
  year = {2023},
  doi = {10.1103/PRXQuantum.4.020355}
}

@article{Schon2005,
  title = {Sequential generation of entangled multiqubit states},
  author = {Sch{\"o}n, C. and Solano, E. and Verstraete, F. and Cirac, J. I. and Wolf, M. M.},
  journal = {Phys. Rev. Lett.},
  volume = {95},
  pages = {110503},
  year = {2005},
  doi = {10.1103/PhysRevLett.95.110503}
}

@article{Schon2007,
  title = {Sequential generation of matrix-product states in cavity QED},
  author = {Sch{\"o}n, C. and Hammerer, K. and Wolf, M. M. and Cirac, J. I. and Solano, E.},
  journal = {Phys. Rev. A},
  volume = {75},
  pages = {032311},
  year = {2007},
  doi = {10.1103/PhysRevA.75.032311}
}

@article{LindnerRudolph2009,
  title = {Proposal for pulsed on-demand sources of photonic cluster state strings},
  author = {Lindner, Netanel H. and Rudolph, Terry},
  journal = {Phys. Rev. Lett.},
  volume = {103},
  pages = {113602},
  year = {2009},
  doi = {10.1103/PhysRevLett.103.113602}
}

@article{Li2022,
  title = {Photonic resource state generation from a minimal number of quantum emitters},
  author = {Li, Bikun and Economou, Sophia E. and Barnes, Edwin},
  journal = {npj Quantum Inf.},
  volume = {8},
  pages = {11},
  year = {2022},
  doi = {10.1038/s41534-022-00522-6}
}

@article{GKP2001,
  title = {Encoding a qubit in an oscillator},
  author = {Gottesman, Daniel and Kitaev, Alexei and Preskill, John},
  journal = {Phys. Rev. A},
  volume = {64},
  pages = {012310},
  year = {2001},
  doi = {10.1103/PhysRevA.64.012310}
}

@article{Menicucci2006,
  title = {Universal quantum computation with continuous-variable cluster states},
  author = {Menicucci, Nicolas C. and van Loock, Peter and Gu, Mile and Weedbrook, Christian and Ralph, Timothy C. and Nielsen, Michael A.},
  journal = {Phys. Rev. Lett.},
  volume = {97},
  pages = {110501},
  year = {2006},
  doi = {10.1103/PhysRevLett.97.110501}
}

@article{Menicucci2014,
  title = {Fault-tolerant measurement-based quantum computing with continuous-variable cluster states},
  author = {Menicucci, Nicolas C.},
  journal = {Phys. Rev. Lett.},
  volume = {112},
  pages = {120504},
  year = {2014},
  doi = {10.1103/PhysRevLett.112.120504}
}

@article{Noh2020,
  title = {Encoding an oscillator into many oscillators},
  author = {Noh, Kyungjoo and Girvin, S. M. and Jiang, Liang},
  journal = {Phys. Rev. Lett.},
  volume = {125},
  pages = {080503},
  year = {2020},
  doi = {10.1103/PhysRevLett.125.080503}
}

@article{Heeres2017,
  title = {Implementing a universal gate set on a logical qubit encoded in an oscillator},
  author = {Heeres, Reinier W. and Reinhold, Philip and Ofek, Nissim and Frunzio, Luigi and Jiang, Liang and Devoret, Michel H. and Schoelkopf, Robert J.},
  journal = {Nat. Commun.},
  volume = {8},
  pages = {94},
  year = {2017},
  doi = {10.1038/s41467-017-00045-1}
}

@article{Pichler2017,
  title = {Universal photonic quantum computation via time-delayed feedback},
  author = {Pichler, Hannes and Choi, Soonwon and Zoller, Peter and Lukin, Mikhail D.},
  journal = {Proc. Natl. Acad. Sci. U.S.A.},
  volume = {114},
  pages = {11362--11367},
  year = {2017},
  doi = {10.1073/pnas.1711003114}
}

@article{Adesso2006,
  title = {Entanglement in {G}aussian matrix-product states},
  author = {Adesso, Gerardo and Ericsson, Marie},
  journal = {Phys. Rev. A},
  volume = {74},
  pages = {030305},
  year = {2006},
  doi = {10.1103/PhysRevA.74.030305}
}

@article{Botero2003,
  title = {Modewise entanglement of {G}aussian states},
  author = {Botero, Alonso and Reznik, Benni},
  journal = {Phys. Rev. A},
  volume = {67},
  pages = {052311},
  year = {2003},
  doi = {10.1103/PhysRevA.67.052311}
}

@article{Eisert2010,
  title = {Colloquium: Area laws for the entanglement entropy},
  author = {Eisert, J. and Cramer, M. and Plenio, M. B.},
  journal = {Rev. Mod. Phys.},
  volume = {82},
  pages = {277--306},
  year = {2010},
  doi = {10.1103/RevModPhys.82.277}
}

@article{Reck1994,
  title = {Experimental realization of any discrete unitary operator},
  author = {Reck, Michael and Zeilinger, Anton and Bernstein, Herbert J. and Bertani, Philip},
  journal = {Phys. Rev. Lett.},
  volume = {73},
  pages = {58--61},
  year = {1994},
  doi = {10.1103/PhysRevLett.73.58}
}

@article{Clements2016,
  title = {Optimal design for universal multiport interferometers},
  author = {Clements, William R. and Humphreys, Peter C. and Metcalf, Benjamin J. and Kolthammer, W. Steven and Walmsley, Ian A.},
  journal = {Optica},
  volume = {3},
  pages = {1460--1465},
  year = {2016},
  doi = {10.1364/OPTICA.3.001460}
}

@article{Maslov2008,
  title = {Quantum circuit simplification and level compaction},
  author = {Maslov, Dmitri and Dueck, Gerhard W. and Miller, D. Michael and Negrevergne, Camille},
  journal = {IEEE Trans. Comput.-Aided Des. Integr. Circuits Syst.},
  volume = {27},
  pages = {436--444},
  year = {2008},
  doi = {10.1109/TCAD.2007.911334}
}

@article{Nam2018,
  title = {Automated optimization of large quantum circuits with continuous parameters},
  author = {Nam, Yunseong and Ross, Neil J. and Su, Yuan and Childs, Andrew M. and Maslov, Dmitri},
  journal = {npj Quantum Inf.},
  volume = {4},
  pages = {23},
  year = {2018},
  doi = {10.1038/s41534-018-0072-4}
}

@article{Wu2023,
  title = {Optimal encoding of oscillators into more oscillators},
  author = {Wu, Jing and Brady, Anthony J. and Zhuang, Quntao},
  journal = {Quantum},
  volume = {7},
  pages = {1082},
  year = {2023},
  publisher = {Verein zur F{\"o}rderung des Open Access Publizierens in den Quantenwissenschaften},
  doi = {10.22331/q-2023-08-16-1082}
}

@article{Brady2024,
  title = {Safeguarding Oscillators and Qudits with Distributed Two-Mode Squeezing},
  author = {Brady, Anthony J. and Wu, Jing and Zhuang, Quntao},
  journal = {Quantum},
  volume = {8},
  pages = {1478},
  year = {2024},
  publisher = {Verein zur F{\"o}rderung des Open Access Publizierens in den Quantenwissenschaften},
  doi = {10.22331/q-2024-09-19-1478}
}

@article{Shor1995,
  title = {Scheme for reducing decoherence in quantum computer memory},
  author = {Shor, Peter W.},
  journal = {Phys. Rev. A},
  volume = {52},
  pages = {R2493},
  year = {1995},
  doi = {10.1103/PhysRevA.52.R2493}
}

@article{Bellman1962,
  title = {Dynamic programming treatment of the travelling salesman problem},
  author = {Bellman, Richard},
  journal = {J. Assoc. Comput. Mach.},
  volume = {9},
  pages = {61},
  year = {1962},
  doi = {10.1145/321105.321111}
}

@article{HeldKarp1962,
  title = {A dynamic programming approach to sequencing problems},
  author = {Held, Michael and Karp, Richard M.},
  journal = {J. Soc. Indust. Appl. Math.},
  volume = {10},
  pages = {196},
  year = {1962},
  doi = {10.1137/0110015}
}

@article{Asavanant2019,
  title = {Generation of time-domain-multiplexed two-dimensional cluster state},
  author = {Asavanant, Warit and Shiozawa, Yu and Yokoyama, Shota and Charoensombutamon, Baramee and Emura, Hiroki and Alexander, Rafael N. and Takeda, Shuntaro and Yoshikawa, Jun-ichi and Menicucci, Nicolas C. and Yonezawa, Hidehiro and Furusawa, Akira},
  journal = {Science},
  volume = {366},
  pages = {373},
  year = {2019},
  doi = {10.1126/science.aay2645}
}

@article{Larsen2019,
  title = {Deterministic generation of a two-dimensional cluster state},
  author = {Larsen, Mikkel V. and Guo, Xueshi and Breum, Casper R. and Neergaard-Nielsen, Jonas S. and Andersen, Ulrik L.},
  journal = {Science},
  volume = {366},
  pages = {369},
  year = {2019},
  doi = {10.1126/science.aay4354}
}

@article{Yokoyama2013,
  title = {Ultra-large-scale continuous-variable cluster states multiplexed in the time domain},
  author = {Yokoyama, Shota and Ukai, Ryuji and Armstrong, Seiji C. and Sornphiphatphong, Chanond and Kaji, Toshiyuki and Suzuki, Shigenari and Yoshikawa, Jun-ichi and Yonezawa, Hidehiro and Menicucci, Nicolas C. and Furusawa, Akira},
  journal = {Nat. Photonics},
  volume = {7},
  pages = {982},
  year = {2013},
  doi = {10.1038/nphoton.2013.287}
}

@article{Yoshikawa2016,
  title = {Invited Article: Generation of one-million-mode continuous-variable cluster state by unlimited time-domain multiplexing},
  author = {Yoshikawa, Jun-ichi and Yokoyama, Shota and Kaji, Toshiyuki and Sornphiphatphong, Chanond and Shiozawa, Yu and Makino, Kenzo and Furusawa, Akira},
  journal = {APL Photonics},
  volume = {1},
  pages = {060801},
  year = {2016},
  doi = {10.1063/1.4962732}
}

@article{Chakhmakhchyan2018,
  title = {Simulating arbitrary Gaussian circuits with linear optics},
  author = {Chakhmakhchyan, Levon and Cerf, Nicolas J.},
  journal = {Phys. Rev. A},
  volume = {98},
  pages = {062314},
  year = {2018},
  doi = {10.1103/PhysRevA.98.062314}
}

@article{Yin2013,
  title = {Catch and release of microwave photon states},
  author = {Yin, Yi and Chen, Yu and Sank, Daniel and O'Malley, P. J. J.
            and White, T. C. and Barends, R. and Kelly, J. and Lucero, Erik
            and Mariantoni, Matteo and Megrant, A. and Neill, C.
            and Vainsencher, A. and Wenner, J.
            and Korotkov, Alexander N. and Cleland, A. N.
            and Martinis, John M.},
  journal = {Phys. Rev. Lett.},
  volume = {110},
  pages = {107001},
  year = {2013},
  doi = {10.1103/PhysRevLett.110.107001}
}

@article{Pierre2014,
  title = {Storage and on-demand release of microwaves using superconducting resonators with tunable coupling},
  author = {Pierre, Mathieu and Svensson, Ida-Maria
            and Sathyamoorthy, Sankar Raman and Johansson, G{\"o}ran
            and Delsing, Per},
  journal = {Appl. Phys. Lett.},
  volume = {104},
  pages = {232604},
  year = {2014},
  doi = {10.1063/1.4882646}
}

@article{Houck2007,
  title = {Generating single microwave photons in a circuit},
  author = {Houck, A. A. and Schuster, D. I. and Gambetta, J. M.
            and Schreier, J. A. and Johnson, B. R. and Chow, J. M.
            and Frunzio, L. and Majer, J. and Devoret, M. H.
            and Girvin, S. M. and Schoelkopf, R. J.},
  journal = {Nature},
  volume = {449},
  pages = {328--331},
  year = {2007},
  doi = {10.1038/nature06126}
}

@article{Pechal2014,
  title = {Microwave-controlled generation of shaped single photons in circuit quantum electrodynamics},
  author = {Pechal, M. and Huthmacher, L. and Eichler, C.
            and Zeytino{\u{g}}lu, S. and Abdumalikov, A. A., Jr.
            and Berger, S. and Wallraff, A. and Filipp, S.},
  journal = {Phys. Rev. X},
  volume = {4},
  pages = {041010},
  year = {2014},
  doi = {10.1103/PhysRevX.4.041010}
}

@article{Zhang2006,
  title = {Continuous-variable Gaussian analog of cluster states},
  author = {Zhang, Jing and Braunstein, Samuel L.},
  journal = {Phys. Rev. A},
  volume = {73},
  pages = {032318},
  year = {2006},
  doi = {10.1103/PhysRevA.73.032318}
}

@article{Chaitin1982,
  title = {Register allocation \& spilling via graph coloring},
  author = {Chaitin, Gregory J.},
  journal = {ACM SIGPLAN Not.},
  volume = {17(6)},
  pages = {98--101},
  year = {1982},
  doi = {10.1145/872726.806984}
}

@article{Menicucci2007,
  title = {Ultracompact generation of continuous-variable cluster states},
  author = {Menicucci, Nicolas C. and Flammia, Steven T. and Zaidi, Hussain
            and Pfister, Olivier},
  journal = {Phys. Rev. A},
  volume = {76},
  pages = {010302(R)},
  year = {2007},
  doi = {10.1103/PhysRevA.76.010302}
}

\end{document}